%% file: main_arxiv.tex
\documentclass[journal, 10pt]{IEEEtran}

\usepackage{comment}
\usepackage{microtype}
\usepackage{hhline}
\usepackage{amsfonts}
\usepackage{algorithmic}
\usepackage{algorithm}
\usepackage{amsthm}
\usepackage{cite}
\usepackage[cmex10]{amsmath}
\usepackage{amssymb}
\usepackage{graphicx}
\graphicspath{{./results/}}
\usepackage{cite}
\usepackage{psfrag}
\usepackage{siunitx}
\usepackage{acronym}
\usepackage{balance}
\usepackage{color}
\usepackage{trfsigns}
\usepackage{stfloats}
\usepackage{mathtools}
\usepackage{multirow}

\usepackage{enumitem}
\usepackage{ifthen}
\usepackage{tablefootnote}
\usepackage[dvipsnames,table]{xcolor}
\usepackage[hidelinks]{hyperref}
\usepackage{array,multirow,graphics}

\definecolor{mygreen}{rgb}{0.0941 0.5804 0.0941}
\definecolor{myblue}{rgb}{0 0.8078 0.8078}
\definecolor{myred}{rgb}{0.8549 0.4863 0.4863}

\newtheorem{conjecture}{Conjecture}
\input{figures/tikz_def}

\newlength{\eqboxstorage}

\ifCLASSOPTIONcompsoc
\usepackage[caption=false,font=normalsize,labelfon
t=sf,textfont=sf]{subfig}
\else
\usepackage[caption=false,font=footnotesize]{subfi
g}
\fi

\newcommand{\fdl}{{f}_{\mathrm{d},k}}
\newcommand{\norm}[1]{\lVert#1\rVert}
\newcommand{\fd}{{f}_{\mathrm{d}}}
\newcommand{\dtx}{d_{\mathrm{t}}}
\newcommand{\drx}{d_{\mathrm{r}}}
\newcommand{\htx}{h_{\mathrm{t}}}
\newcommand{\hrx}{h_{\mathrm{r}}}
\newcommand{\fc}{f_{\mathrm{c}}}

\newcommand{\vtxvec}{\mathbf{v}_{\mathrm{t}}}
\newcommand{\vrxvec}{\mathbf{v}_{\mathrm{r}}}
\newcommand{\vtxx}{{v}_{\mathrm{t}x}}
\newcommand{\vrxx}{{v}_{\mathrm{r}x}}
\newcommand{\vtxy}{{v}_{\mathrm{t}y}}
\newcommand{\vrxy}{{v}_{\mathrm{r}y}}
\newcommand{\vtxz}{{v}_{\mathrm{t}z}}
\newcommand{\vrxz}{{v}_{\mathrm{r}z}}
\newcommand{\tr}{^{\mathrm{T}}}

\newcounter{MYtempeqncnt}

\acrodef{WSSUS}{wide-sense stationary uncorrelated scattering}
\acrodef{LSF}{local scattering function}
\acrodef{pdf}{probability density function}
\acrodef{PDP}{power delay profile}
\acrodef{SR}{specular reflection}
\acrodef{LOS}{line-of-sight}
\acrodef{PDF}{probability density function}
\acrodef{V2V}{vehicle-to-vehicle}
\acrodef{M2M}{mobile-to-mobile}
\acrodef{A2A}{air-to-air}
\acrodef{F2M}{fixed-to-mobile}
\acrodef{PSC}{prolate spheroidal coordinate}
\acrodef{PSCS}{prolate spheroidal coordinate system}
\acrodef{CCS}{Cartesian coordinate system}
\acrodef{TX}{transmitter}
\acrodef{RX}{receiver}
\acrodef{WSS}{wide-sense stationary}
\acrodef{US}{uncorrelated scattering}
\acrodef{CCF}{channel correlation function}
\acrodef{OTFS}{orthogonal time frequency space}
\acrodef{ISAC}{integrated sensing and communications}

\begin{document}

\IEEEoverridecommandlockouts
\IEEEpubid{\scriptsize This work is planned to be submitted to the IEEE for possible publication. Copyright may be transferred without notice, after which this version may no longer be accessible.}

\title{Mobile-to-Mobile Uncorrelated Scatter Channels}
\author{Michael~Walter, Martin~Schmidhammer, Miguel A. Bellido-Manganell, Thomas~Wiedemann, and Dmitriy~Shutin
\thanks{M. Walter, M. Schmidhammer, M. A. Bellido-Manganell, T. Wiedemann, and D. Shutin are with the German Aerospace Center (DLR), Institute of Communications and Navigation, 82234 Wessling, Germany (e-mail: \{m.walter, martin.schmidhammer, miguel.bellidomanganell, thomas.wiedemann, dmitriy.shutin\}@dlr.de).}}

\maketitle

\begin{abstract}
In this paper, we present a complete analytic probability based description of mobile-to-mobile uncorrelated scatter channels.
The correlation based description introduced by Bello and Matz is thus complemented by the presented probabilistic description leading to a common theoretical description of uncorrelated scatter channels.
Furthermore, we introduce novel two-dimensional hybrid characteristic probability density functions, which remain a probability density in one of the variables and a characteristic function in the other variable. 
Such a probability based description allows us to derive a mathematical model, in which the attenuation of the scattering components is inherently included in these two-dimensional functions.
Therefore, there is no need to determine the path loss exponent. 
Additionally, the Doppler probability density function with the inclusion of the path loss leads to a concave function of the Doppler spectrum, which is quite different from the Jakes and Doppler spectra and can be directly parameterized by the velocity vectors and geometry of the scattering plane. 
Thus, knowing those parameters permits the theoretical computation of the Doppler spectra and temporal characteristic functions.
Finally, we present a comparison between the computed probability based theoretical results and measurement data for a generic mobile-to-mobile channel.
The agreement between the two shows the usefulness of the probability based description and confirms new shapes of the Doppler power spectra.
\end{abstract}


\section{Introduction}
\IEEEPARstart{T}{he} significance of \ac{M2M} communication is progressively growing. 
Especially \ac{V2V} communications \--- just one area of modern \ac{M2M} communications \--- is becoming increasingly integrated into newly manufactured cars to mitigate possible accidents via situational awareness.
Moreover, it will become an integral part of future autonomous driving vehicle solutions \cite{Darbha19}.    
Similar direct communication frameworks are envisioned for various modes of transportation such as trains, ships, aircraft, and drones. 
Historically, the channel models used for the design and testing of communication systems in the past were \ac{WSSUS}.
They were largely valid as they were mainly for fixed-to-mobile channels.
However, they are inadequate in an \ac{M2M} context due to the inherent non-stationarity and uncorrelated scattering characteristics of the channel.

Scattering is ultimately a stochastic process.
It has been shown in \cite{Rice44} and \cite{Rice45} that it can be well explained by means of postulating a power spectral density in the frequency domain, or equivalently with an autocorrelation function in the temporal domain.  
Stochastic time variations can also be described in a similar statistical fashion. \IEEEpubidadjcol
Clarke derived in \cite{Clarke68} a model that describes the Doppler spectrum of the propagation channel in the case of a mobile receiver, e.g., a moving car communicating with a fixed base station. 
This type of Doppler spectrum is widely known as Jakes spectrum \cite{jakes94}. 
In \cite{Bello63} Bello extended the mentioned models to more general channels, describing statistical characterizations for both \ac{WSSUS} and non-\ac{WSSUS} channels. 
Later, Matz in \cite{Matz05} focuses on non-\ac{WSSUS} channels, which he studies from two complementary perspectives.
First, the time-frequency channel transfer function can be treated as a non-stationary process in both time and frequency domains.
Second, by assuming that scatterers with distinct delays and Doppler frequencies are uncorrelated, the resulting channel impulse response can be studied as non-stationary in time and uncorrelated along the delay.
This allows to introduce the \ac{LSF} and the \ac{CCF}, which both describe small-scale channel statistics.
The two functions are naturally related to correlation functions introduced by Bello in \cite{Bello63} via Fourier transforms.

Bello states in Section~IV of his seminal work \cite{Bello63} that it's difficult to find an exact statistical description of a time-variant channel in terms of multidimensional probability functions and that correlation functions are a more practical approach.
In essence, this statement implies that without further physical assumptions on the propagation environment, an accurate statistical characterization of non-\ac{WSSUS} channels can be quite elusive. 
This motivates the application of geometric-stochastic modeling approaches, where a specific propagation geometry is assumed, or at least some additional assumptions on the structure of the channel are put into place.

One such approach, as exemplified well in \cite{Bernado12}, reveals that \ac{M2M} channels typically violate the \ac{WSS} assumption to a greater extent than the \ac{US} assumption. 
This, on the one hand, constrains general non-\ac{WSSUS} type of channels to more restricted cases, while on the other hand, building a highly relevant application scenario for the design of practical \ac{M2M} communication systems.
Our objective in this paper is therefore to explore such restricted cases, especially non-\ac{WSS} channels, in more detail, providing stochastic characterizations describing the non-stationary behavior of such \ac{M2M} channels.
Thus, we strive to extend our theoretical non-\ac{WSS} model to encompass other correlation domains in an attempt to create a complete stochastic description of a non-stationary \ac{M2M} channel characterization.
Therefore, we use a \ac{PSC} system introduced in \cite{Walter_AWPL14} to derive the time-dependent \acp{pdf} of delay and Doppler variables.
Additionally, we introduce hybrid characteristic and probability density functions to better capture the complex dynamics of \ac{M2M} communication channels.
These are in essence (inverse) Fourier transforms along one of the variables of the corresponding time-varying \ac{pdf}.
The newly introduced hybrid functions are proportional to the correlation functions introduced by Bello.
We thus extend the proof of the proportionality between the joint delay Doppler \ac{pdf} and scattering function for the \ac{WSSUS} case \cite{Walter_TVT17} to non-\ac{WSS} channels.
In particular, the proportionality between the joint time frequency characteristic function and the time frequency correlation function is shown. 
The time-variant time frequency correlation function shown in \cite{Gutierrez19} by invoking the \ac{US} assumption is thus proportional to the joint characteristic function presented in this paper. 
We further explain new insights into the channel structure obtained from newly introduced hybrid channel functions. 
By being able to theoretically describe high mobility \ac{M2M} channels, new modulation schemes like \ac{OTFS} \cite{Zhiqiang21}, \cite{Shuangyang21} can be employed for \ac{ISAC} \cite{Keskin24}.

The rest of the paper is structured as follows. 
In Section~II, we provide a complete stochastic description of the uncorrelated \ac{M2M} channel and relate those functions to the correlation functions, \ac{LSF}, and \ac{CCF}. 
In Section~III, we show, how the hybrid probability characteristic function is calculated in closed form. 
The theoretical stochastic description is then verified by measurement data from an air-to-air measurement campaign in Section~IV. The paper is concluded with Section~V. 

\section{Characterization of US Channels}\label{sec:CH_Characterization}
\noindent
In order to obtain analytical solutions for the probability based description, it is essential to establish a unified mathematical framework.
Our goal is to establish connections between the models developed  by Bello in \cite{Bello63} and Matz in \cite{Matz05} with our probabilistic approach.

Consider a classical communication channel between a transmitter and a receiver, both possibly mobile. 
The relationship between the transmitted signal $s(t)$ and the received signal $r(t)$ can be represented as \cite{Matz05}
\begin{IEEEeqnarray}{rCl}
    r(t)&=&\int h(t,\tau)s(t-\tau)\,\mathrm{d}\tau\,,
    \label{eq:channel_response}
\end{IEEEeqnarray}
where the function $h(t,\tau)$ \--- the time-varying channel impulse response \--- fully characterizes the propagation environment between the transmitter and the receiver. 
Essentially \eqref{eq:channel_response} states that the received signal is a superposition of differentially delayed copies of the transmitted signal.
The channel can be defined in terms of the spreading function or also known as Doppler-variant impulse response $S(\fd,\tau)=\int h(t,\tau)\mathrm{e}^{-\mathrm{j}2\pi\fd t}\,\mathrm{d}t$, the time-varying transfer function $L(t,f)=\int h(t,\tau)\mathrm{e}^{-\mathrm{j}2\pi f \tau}\,\mathrm{d}\tau$, via a Fourier transform over the delay variable $\tau$ or the Doppler-variant transfer function $T(\fd,f)=\int\int h(t,\tau)\mathrm{e}^{-\mathrm{j}2\pi\fd t}\mathrm{e}^{-\mathrm{j}2\pi f \tau}\,\mathrm{d}t\mathrm{d}\tau$, which is obtained by a double Fourier transform of $h(t,\tau)$ over the $t$ and $\tau$ variables.
Knowledge of $h(t,\tau)$, or any of the other three functions is thus instrumental for the design, simulation, or testing of practical communication systems.

\subsection{Correlation Based Description}
\noindent
In practice, an exact form of $h(t,\tau)$ depends on the particular propagation environment.
The environment, however, is rarely known accurately or in advance at the stage of communication system design. 
Therefore, as has been mentioned earlier, statistical properties of $h(t,\tau)$ are of interest.
Similarly, the statistical properties of the other three system functions $S(\fd,\tau)$, $L(t,f)$, and $T(\fd,f)$ can be determined.
These statistics can be captured by the corresponding autocorrelation functions 
of the four system functions according to \cite{Bello63, Matz05, Boashash15} as
\begin{IEEEeqnarray}{lll}
    \label{eq:channel_corr_func}
    R_h(t,\tau;\Delta t,\Delta \tau)&=\mathbb{E}\{h(t,\tau+\Delta\tau)h^*(t-\Delta t,\tau)\}\,,\\
    R_L(t,f;\Delta t,\Delta f)&=\mathbb{E}\{L(t ,f+\Delta f)L^*(t-\Delta t, f)\}\,,\nonumber\\
    R_S(\tau,\fd;\Delta \tau,\Delta \fd)&=\mathbb{E}\{S(\tau,\fd +\Delta \fd)S^*(\tau-\Delta\tau,\fd)\}\,,\nonumber\\
    R_T(\fd,f;\Delta \fd,\Delta f)&=\mathbb{E}\{T(\fd,f+ \Delta f)T^*(\fd-\Delta\fd, f)\}\,.\nonumber    
\end{IEEEeqnarray}
Here the operator $\mathbb{E}\{\cdot\}$ denotes the expectation operation.\footnote{
Due to equivalence $\mathbb{E}\{f(x,y+\Delta y)f^*(x-\Delta  x,y)\}=\mathbb{E}\{f(x,y)f^*(x-\Delta x,y-\Delta y)\}=\mathbb{E}\{f(x+\Delta x,y+\Delta y)f^*(x,y)\}$ expressions in \eqref{eq:channel_corr_func} can be modified correspondingly.}
By taking double Fourier transforms of the correlation functions, a set of four equivalent, yet different descriptions of the channel can be obtained.

One can see that the correlation functions in \eqref{eq:channel_corr_func} are 4D functions.
These functions are difficult to work with, not to mention hard to get insights into or intuition about their properties. 
To simplify the analysis, one can often invoke the \ac{WSSUS} assumption, see e.g., \cite{Bello63}.
Its consequence is that correlation functions in \eqref{eq:channel_corr_func} become dependent only on the corresponding time and frequency lag variables.
They thus collapse to  much simpler 2D functions, see \cite{Bello63} and \cite{Matz05}.    
Due to this simplification, the \ac{WSSUS} assumption has dominated channel modeling over decades, especially for non-mobile applications or \ac{F2M} channels, i.e., with a fixed base station. 
Yet, the mobility of \ac{TX} and \ac{RX} nowadays brings in the non-stationarity of the propagation environment, and thus of the channel.
In \cite{Matz05} the \ac{LSF} was introduced as an alternative second-order channel statistic to account for non-stationarities in the time and frequency domains.
The \ac{LSF} generalizes the scattering function for non-\ac{WSSUS} channels, but, as mentioned above, it is a 4D function.  
Many \ac{M2M} channels, however, in particular \ac{V2V} channels, are non-stationary, yet preserve the \ac{US} property.
This permits describing \ac{US} channels with only three variables instead of four.

According to \cite{Bello63}, \ac{US} channels were observed for troposcatter communication and moon reflections.
Recently, similar effects were also observed for \ac{M2M} channels, where the uncorrelated scattering assumption was validated with measurements, see e.g., \cite{Bernado12}. 
There the author states that the \ac{V2V} channel infringes the \ac{WSS} assumption much stronger than the \ac{US} assumption.
In the \ac{US} case the scatterers can be modeled as a continuum of uncorrelated scatterers according to \cite{Bello63}.
As a result, the general correlation functions in \eqref{eq:channel_corr_func} become independent of $\Delta \tau$, since the correlation between scatterers causing different delays vanishes.
Under the \ac{US} assumption \eqref{eq:channel_corr_func} can be represented, see also \cite[eq. (66)]{Bello63}, as\footnote{Note that equations in \eqref{eq:3Ddescriptions} imply that the right-hand side is independent of $f$ in case of the \ac{US} assumption.}  
\begin{IEEEeqnarray}{lll}
R_h(t,\tau;\Delta t,\Delta \tau)&=P_h(t;\tau,\Delta t)\mathrm{\delta}(\Delta\tau)\,,\label{eq:3Ddescriptions}\\
R_L(t,f;\Delta t,\Delta f)&=R_L(t;\Delta t,\Delta f)\,,\nonumber\\
R_S(\tau,\fd;\Delta \tau,\Delta \fd)&=P_S(\tau,\fd;\Delta \fd)\delta(\Delta\tau)\,,\nonumber\\ 
R_T(\fd,f;\Delta \fd,\Delta f)&=R_T(\fd;\Delta \fd,\Delta f)\,,\nonumber
\end{IEEEeqnarray}
where $\delta(\cdot)$ is the Dirac delta distribution and $P_h(t;\tau,\Delta t)$ and $P_S(\tau,\fd;\Delta\fd)$ are cross-power spectral densities.
Using the former, the \ac{LSF} is defined \cite{Matz03} as 
\begin{equation}
\label{eq:LSF}
\tilde{C}_{\mathbf{H}}(t;\tau,\fd)=\int P_h(t,\tau;\Delta t)\mathrm{e}^{-\mathrm{j}2\pi\Delta t \fd}\,\mathrm{d}\Delta t\,.
\end{equation}
Furthermore, in \cite{Matz05} Matz defines the corresponding channel correlation function, which we simplify here to
\begin{equation}
\label{eq:CCF}
\tilde{\mathcal{A}}_{\mathbf{H}}(\Delta t,\Delta f;\Delta \fd)=\int R_L(t;\Delta t,\Delta f)\mathrm{e}^{-\mathrm{j}2\pi\Delta\fd t}\,\mathrm{d}t\,,
\end{equation}
which brings the total available functions to describe the \ac{US} channel statistically to six, which are given in \eqref{eq:3Ddescriptions}, as well as \eqref{eq:LSF} and \eqref{eq:CCF}. 
The reduction of dimensionality of the six correlation based functions permits simpler visualization and interpretation, which in turn provide valuable insights into the correlation properties of the channel. 

Again, the resulting correlation functions can be studied both in time, time lag, or frequency and frequency lag domains.
As a consequence, for the 3D functions in \eqref{eq:3Ddescriptions} a total of eight equivalent representations can mathematically be established.
However, in \cite{Bello63} Bello discusses only the four proper correlation functions, where the time and frequency variables, as well as the associated lag variables are in the same domain, e.g. either both in the temporal or both in the frequency domain.
Later in \cite{Matz05} Matz proposed the time-variant \ac{LSF} and the channel correlation function, which are obtained by Fourier transforms of the original four correlation functions as shown in \eqref{eq:LSF} and \eqref{eq:CCF}.
As we will see, it is also useful to consider the correlation variables in a mixed temporal frequency domain, thus extending the available six functions from Bello and Matz to a total of eight. 
All these functions are summarized conceptually in  Fig.~\ref{fig:us_pic} for the probability based description.
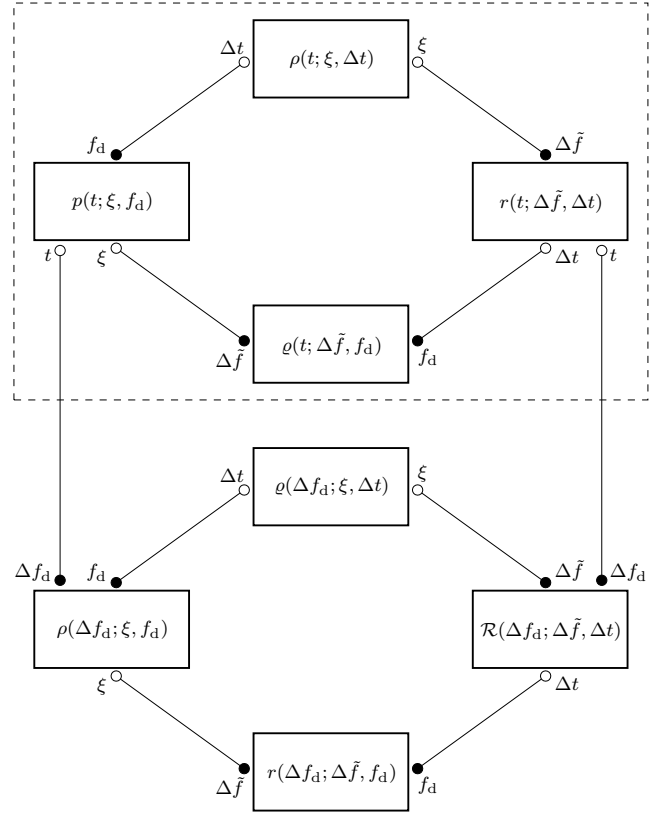
\begin{figure}[tb]
\centering
  \scalebox{0.9}{\input{figures/us_correlation_functions_v3}}
  \caption{Time-variant and Doppler correlated relationships between characteristic functions, probability density functions, and hybrid characteristic probability density functions for the \ac{US} channel. Time-variant functions, which are referred to in Fig.~\ref{fig:correlation_functions}, are framed by a dashed rectangle.}
  \label{fig:us_pic}
\end{figure}

In order to set all correlation based functions in a relationship with the probability based functions discussed later, we need to make a new definition.
Thus, we define 
\begin{equation}
P_{\varrho}(t;\Delta f,\fd)\triangleq\int\tilde{C}_{\mathbf{H}}(t;\tau,\fd)\mathrm{e}^{-\mathrm{j}2\pi\Delta f\tau}\,\mathrm{d}\tau\,,
\end{equation}
as the Fourier transform of the time-variant \ac{LSF} with respect to the delay variable $\tau$.
Similarly to \cite{paetzold2002mobile} and to ease further direct comparison with probability based descriptions defined in the next section, we define normalized versions of the correlation based functions as
\begin{IEEEeqnarray}{lll}
\label{eq:normalized_correlations}
\tilde{C}_{\mathbf{H}}(t;\fd|\tau)&\triangleq\frac{\tilde{C}_{\mathbf{H}}(t;\tau,\fd)}{P_h(t;\tau,\Delta t=0)}\,,\\
P_h(t;\Delta t|\tau)&\triangleq\frac{P_h(t;\tau,\Delta t)}{P_h(t;\tau,\Delta t=0)}\,,\nonumber\\
P_{\varrho}(t;\Delta f|\fd)&\triangleq\frac{P_{\varrho}(t;\Delta f,\fd)}{P_{\varrho}(t;\Delta f=0,\fd)}\,,\nonumber\\
\tilde{R}_L(t;\Delta t,\Delta f)&\triangleq\frac{R_L(t;\Delta t,\Delta f)}
{R_L(t;\Delta t=0,\Delta f=0)}\,,\nonumber\\
\tilde{P}_h(t;\tau)&\triangleq\frac{P_h(t;\tau,\Delta t=0)}{\int P_h(t;\tau,\Delta t=0)\,\mathrm{d}\tau}\,,\nonumber\\
\tilde{P}_{\varrho}(t;\fd)&\triangleq\frac{P_{\varrho}(t;\Delta f=0,\fd)}{\int P_{\varrho}(t;\Delta f=0,\fd)\,\mathrm{d}\fd}\,.\nonumber
\end{IEEEeqnarray}

\subsection{Probability Based Description}
\label{sec:Prob_based}
\noindent
In this subsection, we discuss the probabilistic representation of the \ac{US} channel more formally.
We begin by defining $\xi\triangleq\tau/\tau_{\mathrm{los}}$ as the normalized delay, where $\tau_{\mathrm{los}}$ is the \ac{LOS} delay between the transmitter and receiver.
Similarly, we define the normalized frequency lag $\Delta \tilde{f}\triangleq\Delta f\tau_{\mathrm{los}}$.
In the following, we will largely utilize the notation of our previous works, e.g.,  \cite{Walter_TVT21}, where the focus was on deriving the joint delay Doppler \ac{pdf} $p(t;\xi,\fd)$ for the \ac{M2M} channel.

Note that by taking Fourier transforms along some of the variables of these functions, we obtain equivalent representations, yet in the corresponding frequency domains.
This will allow for other, often quite revealing, interpretations of the channel properties, as will be shown later.    
These different transforms are shown for the \ac{US} case in Fig.~\ref{fig:us_pic}. 
Of particular interest is the time-varying joint delay Doppler \ac{pdf} $p(t;\xi,\fd)$ that we introduced before in \cite{Walter_TVT21}.
The time-varying joint delay Doppler \ac{pdf} $p(t;\xi,\fd)$ can be computed analytically, e.g., for \ac{M2M} \ac{US} channels.
Furthermore, it has a finite support, since the velocity and sensitivity limit the possible delays and Doppler frequencies.
Since the trajectories of \ac{TX} and \ac{RX} are time-variant, it makes sense to use a time-variant probability density.
In our case the time $t$ is a deterministic variable, whereas the delay $\xi$ and Doppler frequency $\fd$ are treated as stochastic variables.
Their distributions are obtained from the assumption that scatterers are uniformly distributed on the ground.
This is followed by a variable transform from the spatial to the Doppler domain, see \cite{Walter_TVT21}.
To illustrate the connection between the correlation functions from \cite{Bello63} and their equivalent forms in \cite{Matz05}, we summarize them in the Table~\ref{tab:variable_names}, using the original notation.\footnote{Note that the variable $\xi$ is defined in \cite{Bello63} as time difference, whereas in this paper it refers to a normalized delay.} 

Our objective is to relate the joint delay Doppler \ac{pdf} $p(t;\xi,\fd)$ to the \acs{LSF} \cite{Matz05} and, what we call, the hybrid characteristic \ac{pdf} $\rho(t;\xi,\Delta t)$ to the original correlation function of Bello \cite{Bello63}.
Recall that for the joint \acp{pdf} this is equivalent to computing the characteristic function for one of the two variables.
Thus, $\rho(t;\xi,\Delta t)$ is a characteristic function along $\Delta t$ direction and a \ac{pdf} function along the delay direction.
Accordingly, this function presents a time-variant spectral density along the $\xi$ variable and the temporal correlation along the $\Delta t$ variable. 
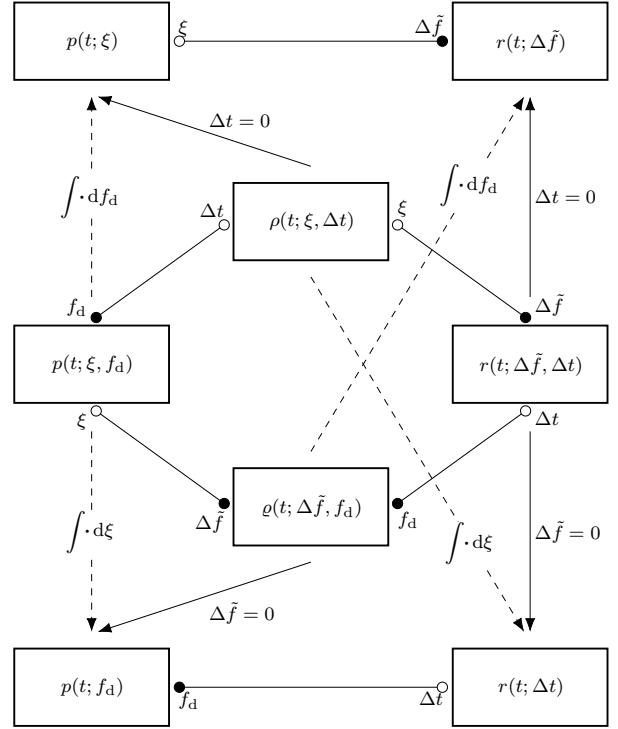
\begin{figure}[tb]
\centering
  \scalebox{0.9}{\input{figures/correlation_functions_v2}}
\caption{Time-variant relationships between characteristic functions, probability density functions, and hybrid probability density characteristic function for the \ac{US} channel.}  \label{fig:correlation_functions}
\end{figure}
Since \ac{M2M} channels are non-stationary, this function is particularly useful, since it allows observing a time-varying,  delay-dependent temporal correlation of the channel. 
Using Fourier transforms to convert the time or frequency variables will lead to other hybrid representations of the channel, where the function represents a characteristic function along one of the variables and a probability density along the other.
Thus, we consider the Fourier and inverse Fourier transform as basically the same operation, but with different signs. 
In Fig.~\ref{fig:us_pic} this is represented with functions, e.g., $\rho(t;\xi,\Delta t)$ and $\varrho(t;\Delta \tilde{f},\fd)$, which are neither pure characteristic nor probability density function representations.

For an easier comparison to \ac{WSSUS} channels, the eight functions in Fig.~\ref{fig:us_pic} can be partitioned into (i) a time-variant  (dashed box) and (ii) a Doppler correlated description. 
We consider time-variant descriptions as more natural and easier to interpret, although all descriptions enjoy an equivalence under the appropriate Fourier transform.
The reasoning behind this lies in the fact that in \ac{M2M} scenarios both transmitter and receiver move, resulting in time-dependent velocity vectors.
Thus, a time-variant channel description would be more natural.
Yet, via appropriate Fourier transform we can equivalently obtain Doppler correlated descriptions \--- the lower part of Fig.~\ref{fig:us_pic}.
These are, however, less intuitive to interpret.

\subsubsection{Time-variant Functions}
We begin with the known time-varying joint delay Doppler \ac{pdf} $p(t;\xi,\fd)$ from \cite{Walter_TVT21} as a starting point for further analysis.
An inverse Fourier transform along the variables of delay leads to a hybrid characteristic \ac{pdf} representation, as we mentioned above. 
Let us consider the joint delay Doppler \ac{pdf} $p(t;\xi,\fd)$ and the corresponding hybrid representation 
\begin{equation}    
    \rho(t;\xi,\Delta t)\triangleq\int p(t;\xi,\fd)\mathrm{e}^{\mathrm{j}2\pi\fd\Delta t}\,\mathrm {d}\fd\,.
    \label{eq:hybridtimedelay}
\end{equation}
It is important to note that $\rho(t;\xi,\Delta t=0)=p(t;\xi)$, since \eqref{eq:hybridtimedelay} becomes a marginalization integral then.
Indeed, by setting the characteristic variable $\Delta t$ to zero, the exponential function in the integral vanishes, and the marginal, time-varying \ac{pdf} $p(t;\xi)$ can be obtained.
This time-variant delay \ac{pdf} is proportional to $P_h(t;\xi,\Delta t=0)$, which is the power delay profile of the channel, see Conjecture \ref{conjecture}.
Thus, we have a non-parametric, geometry-based, time-variant path loss model, which we will analytically derive for the general \ac{M2M} channel. 
With the time-variant delay \ac{pdf} $p(t;\xi)$ we can obtain the factorization of the time-varying, delay-dependent Doppler probability density as $p(t;\xi,\fd)=p(t;\xi)p(t;\fd|\xi)$, which 
reveals the conditional density $p(t;\fd|\xi)$, i.e., the Doppler \ac{pdf} conditioned on a particular delay $\xi$.

With the time-varying \ac{pdf} $p(t;\xi)$ the hybrid characteristic \ac{pdf} can also be factorized to a delay-dependent characteristic function as 
\begin{equation}    
\rho(t;\Delta t|\xi)=\frac{\rho(t;\xi, \Delta t)}{\rho(t;\xi,\Delta t=0)} =\frac{\rho(t;\xi, \Delta t)}{p(t;\xi)}\,.
\end{equation}

The other hybrid characteristic \ac{pdf} contains a \ac{pdf} in the Doppler domain and a characteristic function in the frequency domain, but here the time and Doppler frequency are in different domains. 
The function is given by
\begin{equation}    
    \varrho(t;\Delta \tilde{f}, \fd)\triangleq\int p(t;\xi,\fd)\mathrm{e}^{-\mathrm{j}2\pi\Delta \tilde{f}\xi}\,\mathrm {d}\xi\,.
    \label{eq:hybriddoppler}
\end{equation}
Similarly, by setting $\Delta \tilde{f}=0$ in \eqref{eq:hybriddoppler}, we obtain $\varrho(t;\Delta \tilde{f}=0,\fd)=p(t;\fd)$, i.e., the time-varying Doppler probability density.
The conditional characteristic function is obtained by
\begin{equation}    
    \varrho(t;\Delta \tilde{f}|\fd)=\frac{\varrho(t;\Delta \tilde{f}, \fd)}{\varrho(t;\Delta \tilde{f}=0,\fd)}=\frac{\varrho(t;\Delta \tilde{f}, \fd)}{p(t;\fd)}\,.
\end{equation}
The time-variant joint characteristic function $r(t; \Delta \tilde{f}, \Delta t)$ can be directly obtained by a double Fourier transform as
\begin{equation} 
\label{eq:joint_char}
    r(t;\Delta \tilde{f}, \Delta t)=\iint p(t;\xi,\fd) \mathrm{e}^{-\mathrm{j}2\pi\left(\Delta \tilde{f}\xi-\fd\Delta t\right)}\,\mathrm {d}\xi\mathrm {d} \fd\,,
\end{equation}
with the property $r(t,\Delta \tilde{f}=0,\Delta t=0)=1$.
Note that in \eqref{eq:joint_char} we use one normal and one inverse Fourier transform to compute the joint characteristic function to be consistent with the channel modeling literature.
This is similar to the description used in\cite{Cohen89}.

\begin{table*}[tb]
	\begin{center}
		\caption{Comparison of correlation and probability based functions for US channels.}
		\label{tab:variable_names}
        \renewcommand{\arraystretch}{1.3}
		\begin{tabular}{|l|l|l||l|l|} \hline
	   \rowcolor{lightgray}\textbf{Function Name} & \textbf{Bello \cite{Bello63}} & \textbf{Matz \cite{Matz05},\cite{Matz03}}  & \textbf{Walter et.al.} & 
        \textbf{Function Name}\\ \hline
        time-varying scattering function & $-$ & $\tilde{\mathcal{C}}_{\mathbf{H}}(t;\tau,\nu)$ & $ p(t;\xi,\fd)$& TV joint delay Doppler \ac{pdf}\\
        \hline
        time correlated delay cross-power spectral density & $P_{{g}}(t,
        s;\xi)$ & $P_h(t,\tau;\Delta t)$ & $\rho(t;\xi,\Delta t)$& TV hybrid time delay char. \ac{pdf}\\
        \hline
        $-$& $-$ & $-$ &  $\varrho(t;\Delta \tilde{f},\fd)$& TV hybrid freq. Doppler char. pdf\\
        \hline
        autocorrelation of time-variant transfer function & $R_{{T}}(\Omega;t,s)$ & $R_L(t;\Delta t,\Delta f)$ & $r(t;\Delta \tilde{f},\Delta t)$&TV joint time frequency char. fun.\\
        \hhline{|=|=|=#=|=|}
        DC delay cross-power spectral density & $P_{{U}}(\xi;\nu,\mu)$ & $P_S(\tau,\nu;\Delta \nu)$ & $\rho(\Delta\fd;\xi,\fd)$& DC hybrid Doppler delay char. pdf\\
        \hline
        autocorrelation of output Doppler spread function & $R_{{G}}(\Omega,\nu,\mu)$ & $-$ & $r(\Delta\fd;\Delta \tilde{f},\fd)$& DC frequency Doppler char. fun.\\
        \hline
        $-$& $-$ & $-$ & $\varrho(\Delta \fd; \xi,\Delta t)$& DC hybrid time delay char. pdf\\
        \hline
        uncorrelated scatter channel correlation function & $-$ & $\tilde{\mathcal{A}}_{\mathbf{H}}(\Delta t;\Delta f,\Delta \nu)$ & $\mathcal{R}(\Delta \fd;\Delta \tilde{f},\Delta t)$& DC time frequency char. fun.\\
        \hline
        \end{tabular}  
        \renewcommand{\arraystretch}{1}
        \hfill\vspace{1ex}
        TV: time-variant, DC: Doppler correlated
	\end{center}
\end{table*}

\subsubsection{Time-variant First and Second Order Moments}
The time-variant mean delay and delay spread can be easily calculated by using the hybrid characteristic \ac{pdf} and setting $\Delta t=0$. 
This results in the first two delay moments as 
\begin{IEEEeqnarray}{lll}
\label{eq:mu_xi2}
\mu_{\xi}(t)=\int\limits_{\xi_{\mathrm{min}}}^{\xi_{\mathrm{max}}}\xi\rho(t;\xi,\Delta t=0)\,\mathrm{d}\xi\,,\\
\label{eq:sigma_xi2}
\sigma_{\xi}(t)=
\sqrt{\int\limits_{\xi_{\mathrm{min}}}^{\xi_{\mathrm{max}}}\left(\xi-\mu_{\xi}(t)\right)^2\rho(t;\xi,\Delta t=0)\,\mathrm{d}\xi}\,
\end{IEEEeqnarray}
with $\xi_{\mathrm{max}}>\xi_{\mathrm{min}}>\xi_{\mathrm{sr}}$ and 
\begin{equation}
\label{eq:xi_sr2}
\xi_{\mathrm{sr}} = \max\left(\sqrt{\frac{A^2 + B^2 +D^2}{A^2 + B^2 + C^2}},1\right)\,,
\end{equation}
being the delay of the specular reflection relative to the line-of-sight delay. 
Parameters $A$, $B$, $C$ and $D$ are the orientation coefficients of the scattering plane. 
Their geometric interpretation and computation will be discussed Section~\ref{sec:HybricFunctions}. 

If the influence of the delay is removed, we obtain the temporal correlation of the channel per delay.
Instead of calculating the total mean Doppler and Doppler spread, we calculate delay-dependent mean Doppler and Doppler spread. 
These can again be obtained from hybrid characteristic \ac{pdf} as
\begin{IEEEeqnarray}{rll}
\label{eq:mu_fd}
\mu_{\fd|\xi}(t&)=\left.\frac{1}{\mathrm{j}2\pi}\frac{\partial}{\partial \Delta t}{\rho}(t;\Delta t|\xi)\right|_{\Delta t=0}\,,\\
\label{eq:sigma_fd}
\sigma_{\fd|\xi}(t&)=\\
\frac{1}{2\pi}&\left.\sqrt{\left(\frac{\partial}{\partial \Delta t}{\rho}(t;\Delta t|\xi)\right)^2-\frac{\partial^2}{\partial \Delta t^2}{\rho}(t;\Delta t|\xi)}\right|_{\Delta t=0}\nonumber\,,
\end{IEEEeqnarray}
where $\rho(t;\Delta t|\xi)$ is the inverse Fourier transform of $p(t;\fd|\xi)$ and is therefore a characteristic function in the $\Delta t$ variable.

\subsubsection{Proportionality between the Correlation and Probability Based Functions}
The following proposition states that stochastic channel descriptions computed based on the joint delay Doppler \ac{pdf}, as shown in Fig.~\ref{fig:us_pic}, are proportional to the corresponding correlation based functions derived by Bello and Matz in their works. 
Note that the same variables are used. 
\begin{conjecture}\label{conjecture}
The autocorrelation of the time-variant transfer function $R_L(t;\Delta t, \Delta \tilde{f})$ is proportional to the time-variant joint time frequency characteristic function $r(t;\Delta \tilde{f},\Delta t)$
    $$R_L(t;\Delta t, \Delta \tilde{f})\propto r(t;\Delta \tilde{f},\Delta t)\,.$$  
Furthermore, due to the linearity of Fourier transforms, the above stated proportionality applies for each pair of correlation and probability based functions as outlined in Table~\ref{tab:variable_names}.
\end{conjecture}
\begin{proof}  
We follow here the steps similar to those in \cite{Walter_TVT17}.
The starting point is the assumption that for the time $t$ the channel can be represented as a linear combination of $K(t)$ propagation paths
\begin{equation}
    h(t,\tau) = \sum_{k=0}^{K(t)-1}\alpha_k(t)\mathrm{e}^{-\mathrm{j}2\pi f_c\tau_k(t)}\delta(t-\tau_k(t))\,.
\end{equation}
Here $\alpha_k(t)$ is the complex path weight, $\tau_k(t)$ is the time-varying propagation delay, and $f_c$ is the carrier frequency.
By taking the Fourier transform over the delay, a time-variant transfer function can be constructed as
\begin{equation}    
L(t,f)=\sum_{k=0}^{K(t)-1}\alpha_k(t)\mathrm{e}^{-\mathrm{j}2\pi (f+f_c)\tau_k(t)}\,.
\end{equation}
We approximate the time-varying channel with a piece-wise linear approximation. 
This is done by assuming that for a moment of time $t=t'$ the time-dependent path propagation delay can be locally, over the interval $\Delta t$, approximated with a MacLauren series. 
Thus, we can represent a time-varying delay as  
\begin{equation}
\label{eq:DelayTaylor}
\tau_k(t)=\sum_{n=0}^{\infty}\left.\frac{1}{n!}\frac{\mathrm{d}\tau_k(t)}{\mathrm{d}t}\right|_{t=t'}t^n
\approx\,\tilde{\tau}_k(t')+t\left.\frac{\mathrm{d}\tau_k(t)}{\mathrm{d} t}\right\vert_{t=t'}\,.
\end{equation}
Note that in general, $\tilde{\tau}_k(t')$ is constant over the assumed interval $\Delta t$.
The derivative $\left.\frac{\mathrm{d}\tau_k(t)}{\mathrm{d} t}\right\vert_{t=t'}$ is also constant over this interval.
These parameters characterize the intercept and local linear trend of the time-varying delay $\tau_k(t)$ at a point $t'$.
They do, however, change with time $t$, yet at a lower rate.
In other words, they are piece-wise constant functions of $t$.
In the following, we will make the dependency of these variables on $t$ explicit, keeping the piece-wise constant nature of these variables in mind.

We note that under the narrow-band assumption, the Doppler frequency $\fd(t)$ can be defined as $\fd(t)\triangleq -\fc\, \frac{\mathrm{d}\tau(t)}{\mathrm{d} t}$, i.e., when all transmitted frequencies experience the same Doppler shift \cite{Gutierrez19}. 
Using \eqref{eq:DelayTaylor} the time-varying transfer function $L(t,f)$ can be approximated as   
\begin{IEEEeqnarray}{lll}
L(t,f)\approx 
\sum_{k=0}^{K(t)-1}\tilde{\alpha}_l(t)\mathrm{e}^{\mathrm{j}2\pi\fdl(t) t}\mathrm{e}^{-\mathrm{j}2\pi f\tilde{\tau}_k(t)}\,,
\end{IEEEeqnarray}
where $ \tilde{\alpha}_k(t)=\alpha_k(t)\mathrm{e}^{-\mathrm{j}2\pi f_c\tilde{\tau}_k(t)}$.
For the correlation function, we get with $\xi=\tau/\tau_{\mathrm{los}}$ as shown in \cite{Walter_TVT17} and \cite{Gutierrez19} the following
\begin{IEEEeqnarray}{lll}
R_L(t,\Delta t,\Delta f\tau_{\mathrm{los}}&)=\\
&|\tilde{\alpha}(t)|^2\mathbb{E}\left\{\mathrm{e}^{\mathrm{j}2\pi\fdl(t) \Delta t}\mathrm{e}^{-\mathrm{j}2\pi \Delta f\tau_{\mathrm{los}}(t) \xi_k(t)}\right\}\nonumber\,,
\end{IEEEeqnarray}
where $\tau_{\mathrm{los}}(t)$ is the piece-wise constant \ac{LOS} delay and the expectation is taken with respect to the joint distribution $p(t;\xi,\fd)$.
The result of the latter is the joint characteristic function $r(t;\Delta \tilde{f},\Delta t)$ as given in \eqref{eq:joint_char}.
Thus, the correlation function $R_L$ is proportional to the joint characteristic function $r(t;\Delta \tilde{f} ,\Delta t)$.
This relationship between correlation function and characteristic functions is similarly shown in \cite{Papoulis02} if complex exponentials are used for the channel representation.
Let us stress again that the proportionality is valid for an arbitrary, but fixed $t=t'$ and as $\Delta t\to 0$.
\end{proof}

\subsubsection{Doppler Correlated Functions}
The Doppler correlated functions in Fig.~\ref{fig:correlation_functions} constitute for $\Delta \fd=0$ a temporal average of the functions in the upper half.
The hybrid Doppler delay characteristic \ac{pdf} $\rho(\Delta\fd;\xi,\fd)$ for example is calculated as
\begin{equation}
\rho(\Delta\fd;\xi,\fd)\triangleq\int p(t;\xi,\fd)\mathrm{e}^{-\mathrm{j}2\pi\Delta\fd t}\,\mathrm{d}t\,,
\end{equation}
with $\rho(\Delta\fd=0;\xi,\fd)$ being the temporal mean of the joint delay Doppler probability density function due to the Fourier properties.
This was already implicitly used in describing \ac{V2V} scenarios such as two cars driving in opposite directions in \cite{Walter_PIMRC15}.
We normalize the functions in such a way, that the joint \ac{pdf} $p(t;\xi,\fd)$ is a time-variant probability density in the variables $\xi$ and $\fd$. 
Thus, the time-variant joint characteristic function $r(t;\Delta f=0,\Delta t=0)=1$.
The lower half functions with $\Delta \fd$ deviate from \acp{pdf} or characteristic functions by a factor of $T=1/\Delta \fd$.
The channel correlation function $\tilde{\mathcal{A}}_{\mathbf{H}}(\Delta t;\Delta f,\Delta \nu)$ by Matz thus corresponds to $\mathcal{R}(\Delta \fd;\Delta \tilde{f},\Delta t)$. 
We will focus our attention in the remaining paper on the time-variant functions since a time-variant joint \ac{pdf} as a basis of the description seems more natural with time-variant trajectories of \ac{TX} and \ac{RX} as input to our model.

Finally, as illustrated in Table~\ref{tab:variable_names}, we note that the probability based functions $r$ and $\rho$ correspond to the correlation functions $R$ and $P$ of Bello, respectively.
Further, the joint \ac{pdf} $p$ corresponds to the time-variant \ac {LSF} $\tilde{\mathcal{C}}_{\mathbf{H}}$ and Doppler correlated time frequency characteristic function $\mathcal{R}$ to the time-variant channel correlation function $\tilde{\mathcal{A}}_{\mathbf{H}}$.
For completeness, we further presented two new hybrid functions $\varrho$ with mixed variables in both time-variant and Doppler correlated domains.

\section{Derivation of Hybrid Characteristic PDF}\label{sec:HybricFunctions}
\noindent
In order to obtain an analytical closed-form solution for the hybrid characteristic \ac{pdf}, we have to transform the spatial coordinates into an adequate coordinate system.
We have shown in \cite{Walter_AWPL14}, \cite{Walter_TWC19}, \cite{Walter_TVT21} that a prolate spheroidal coordinate system is suitable for this purpose. 
The \ac{PSCS} allows for a delay-dependent description of the \ac{M2M} channel by exploiting the symmetry of the channel by an ellipsoid-based delay description.
We shortly summarize the corresponding formal steps from \cite{Walter_TVT21} and introduce the coordinate system.

\subsection{Prolate Spheroidal Coordinates}
\noindent
The transformation between the \ac{CCS} $(x,y,z)$ and the \acp{PSC} $(\xi,\eta,\vartheta)$ is given by the following equations
\begin{IEEEeqnarray}{rCl}
\label{eq:CoordinateTransform2}
x&=&l\sqrt{\left(\xi^2-1\right)\left(1-\eta^2\right)}\cos\vartheta\,,\\
y&=&l\sqrt{\left(\xi^2-1\right)\left(1-\eta^2\right)}\sin\vartheta\,,\nonumber\\
z&=&l\xi\eta\,,\nonumber
\end{IEEEeqnarray}
where $\xi\in[1,\infty)$, $\eta\in[-1,1]$, $\vartheta\in[0,2\pi)$ are the new coordinates and $l$ in \eqref{eq:CoordinateTransform2} is the focus distance of both \ac{TX} and \ac{RX} to the origin of the Cartesian and the prolate spheroidal coordinate system.
The coordinate $\xi$ represents the constant distance or delay, respectively, between \ac{TX} and \ac{RX} via a single-bounce reflection.
Geometrically, this relationship is represented by an ellipsoid.

Consider a scattering plane via which a signal propagates to the receiver.  
An arbitrarily oriented scattering plane is given in Cartesian coordinates as 
\begin{IEEEeqnarray}{rCl}
\label{eq:ground_plane}
Ax+By+Cz=lD\,,
\end{IEEEeqnarray}
where the four parameters $\{A, B, C, D\}\in\mathbb{R}$ determine its orientation in space.
For our purposes we express \eqref{eq:ground_plane} in \acp{PSC}, which results in
\begin{IEEEeqnarray}{rCl}
\label{eq:ground_plane_pscs}
Al\sqrt{(\xi^2-1)(1-\eta^2)}\cos\vartheta&+&Bl\sqrt{(\xi^2-1)(1-\eta^2)}\sin\vartheta\nonumber\\&+&Cl\xi\eta=lD\,.
\end{IEEEeqnarray}

The scattering plane, as any 2D plane embedded in 3D space, can be parameterized by two independent variables in the selected coordinate system.
Our goal is to obtain a parameterization that allows for a closed-form derivation of the hybrid time delay characteristic \ac{pdf}.
Since we want $\xi$ for a delay-dependent description, we can choose either $\eta$ or $\vartheta$ as the second variable.
In fact, we need both the $(\xi,\eta)$ and $(\xi,\vartheta)$ parameterizations to cover all possible scattering planes in 3D space.
Our main parameterization, however, is in $(\xi,\eta)$-coordinates. 
The remaining scattering planes, which cannot be parameterized by $(\xi,\eta)$ since they are orthogonal to the $z$-axis in the local \ac{CCS}, are described by the $(\xi,\vartheta)$-coordinates.
The $(\xi,\vartheta)$ description actually complements the $(\xi,\eta)$ parameterization.
In the following, we refer to these cases as \emph{general case} and \emph{complementary case}, respectively.

\subsection{Spatial Probability Density}
\noindent
In order to obtain the joint delay Doppler \ac{pdf} we restrict our analysis to scatterers lying on the scattering plane. 
We consider scatterers that lie on the portion of the scattering plane circumscribed by the intersection ellipse.
The resulting scattering ellipse can be generally described by the implicit expression $q(\xi,\eta,\vartheta)=0$, which simplifies to $q(\xi,\eta)=0$, if the parameterization is in $(\xi,\eta)$-coordinates or $q(\xi,\vartheta)=0$, if the parameterization is in $(\xi,\vartheta)$-coordinates.
We assume that the scatterers lying within $q(\xi,\eta,\vartheta)=0$ are identical and uniformly distributed.
Thus, the two-dimensional density $\mathbf{s}$ of the scatterers within the scattering ellipse is modeled as
\begin{equation}
\label{eq:Concept:ScattDensity}
p(t,q(\xi,\eta,\vartheta); \mathbf{s})=\frac{1}{\mathcal{Y}}\,,
\end{equation}
where $\mathcal{Y}$ is the equivalent area of the ellipse $q(\xi,\eta,\vartheta)=0$.

The joint delay Doppler \ac{pdf} is then obtained by transforming the distribution of scatterers $\mathbf{s}$ into $(\xi,\fd)$-coordinates using \eqref{eq:doppler2} and rules of probability transformation as
\begin{equation}
\label{eq:Concept:DelayDoppler}
p\left(t,q(\xi,\eta,\vartheta);\xi,\fd\right)=p\left(t,q(\xi,\eta,\vartheta); \mathbf{s}\right)\left|\mathbf{J_{s}}^{-1}\right|\,,
\end{equation}
where $\mathbf{J_{s}}^{-1}$ is the inverse $2\times2$ Jacobian matrix of the variable transformation.

For the joint delay Doppler \ac{pdf}, the transformation from the spatial domain to the Doppler domain, i.e., $s\mapsto\fd$ or $\mathbf{s}\mapsto(\xi,\fd)$ introduces ambiguities in the mapping.
These ambiguities, however, can be resolved by applying the algebraic curve theory to the Doppler frequency description, see also \cite{Walter_TWC19}.
Furthermore, the locations of the extrema and thus the limiting frequencies of the \acp{pdf} can be determined.

In order to use a spatial distribution of the scatterers, we need to calculate the area enclosed by the delay ellipsoid. 
Additionally, we need a weighting function $w(\xi,\eta)$ that takes into account the path loss, which follows from the radar equation \cite{Skolnik08}
\begin{IEEEeqnarray}{lCl}
\label{eq:DelayPDF}
w(\xi,\eta)=\frac{1}{\left(\xi^2-\eta^2\right)^2}\,.
\end{IEEEeqnarray}
This essentially states that the received power is proportional to the squared distances from the scatterer to \ac{TX} and \ac{RX} as $P\propto \left(\dtx^2\drx^2\right)^{-1}$.

For the \emph{general case}, we obtain the following equation for the weighted elliptic area
\begin{IEEEeqnarray}{lCl}
\label{eq:area_gen}
\mathcal{Y}_1=\\
\int\limits_{\xi_{\mathrm{min}}}^{\xi_{\mathrm{max}}}2\int\limits_{\eta_{1}\left(\xi\right)}^{\eta_{2}\left(\xi\right)}\hspace*{-0.1cm}\frac{w\left(\xi,\eta\right)l^2\sqrt{A^2+B^2+C^2}\left(\xi^2-\eta^2\right)\,\mathrm{d}\eta\mathrm{d}\xi}{\sqrt{\left(\xi^2-1\right)\left(1-\eta^2\right)\left(A^2+B^2\right)-\left(D-C \xi \eta\right)^2}}\,,\nonumber
\end{IEEEeqnarray}
where $\xi_{\mathrm{max}}>\xi_{\mathrm{min}}>\xi_{\mathrm{sr}}$ are the minimum and maximum normalized delay, which can be set by the user. 
The other parameters $\eta_2(\xi)>\eta_1(\xi)$ are given by \eqref{eq:eta12}. 

For the \emph{complementary case} we obtain a weighted circular area as a special case of the elliptic area, since the semi-major axis of the delay ellipsoid is orthogonal to the scattering plane.
It is given by
\begin{equation}
\label{eq:area_comp}
\mathcal{Y}_2=\hspace*{-0.3cm}\iint\limits_{q(\xi,\eta,\vartheta)=0}\hspace*{-0.3cm} w\left(\xi,\frac{D}{C\xi}\right)\,\mathrm{d}S_1
=\int\limits_{\xi_{\mathrm{min}}}^{\xi_{\mathrm{max}}}\int\limits_{0}^{2\pi}\frac{l^2\left(\xi-\frac{D^2}{C^2\xi^3}\right)}{\left(\xi^2-\left(\frac{D}{C\xi}\right)^2\right)^2}\,\mathrm{d}\vartheta\mathrm{d}\xi,
\end{equation}
where $\mathrm{d}S_1=\mathrm{d}\vartheta\mathrm{d}\xi$ 
is the differential scattering area  
with $\xi_{\mathrm{max}}>\xi_{\mathrm{min}}>\xi_{\mathrm{sr}}$ defined similarly to the \emph{general case}.

For deriving the hybrid time delay characteristic \ac{pdf}, we use the spatial density of the scatterers instead of transforming the Doppler frequency $\fd$, as was done in previous works, e.g., \cite{Walter_TVT21,Walter_TWC19,Walter_TVT17}.
We thus either transform over the variable $\eta$ for the \emph{general case} or over the variable $\vartheta$ for the \emph{complementary case}.

\begin{figure*}[!b]
\normalsize
\setcounter{MYtempeqncnt}{\value{equation}}
\setcounter{equation}{\value{MYtempeqncnt}+13}
\hrulefill
\begin{equation}
\label{eq:eta12}
\eta_{\mathrm{1},\mathrm{2}}(\xi)=\frac{ D C \xi \pm \sqrt{ D^2 C^2  \xi^2 -
    (A^2 \xi^2 +B^2 \xi^2 + C^2 \xi^2-A^2 - B^2) (A^2  +
     B^2+D^2 - A^2 \xi^2 - B^2 \xi^2)}}{ A^2 \xi^2 + B^2 \xi^2 + C^2 \xi^2-A^2 - B^2}
\end{equation}
\end{figure*}

\begin{figure*}[!b]
\normalsize
\setcounter{equation}{\value{MYtempeqncnt}+14}
\hrulefill
\begin{IEEEeqnarray}{lCl}
\label{eq:doppler_star}
f_{\mathrm{d},i}^{\star}(t;\xi,\eta)=\frac{1}{\left(A^2+B^2\right) \left(\xi^2-\eta^2\right)}\Bigg(\left(D-C\xi\eta\right)\left(A \left(\vrxx \left(\xi+\eta\right)+\vtxx \left(\xi-\eta\right)\right)+B\left(\vrxy \left(\xi+\eta\right)+\vtxy\left(\xi-\eta\right)\right)\right)\nonumber\\
\pm\sqrt{\left(\left(\xi^2-1\right) \left(1-\eta^2\right)\left(A^2+B^2\right) -\left(D-C\xi\eta\right)^2\right) \left(B \left(\vrxx \left(\xi+\eta\right)+\vtxx \left(\xi-\eta\right)\right)-A\left(\vrxy \left(\xi+\eta\right)+\vtxy \left(\xi-\eta\right)\right)\right)^2}\nonumber\\
+\left(A^2+B^2\right) \left(\vrxz\left(\xi\eta-1\right)\left(\xi+\eta\right)+\vtxz\left(\xi\eta+1\right) \left(\xi-\eta\right)\right)\Bigg)\frac{\fc}{c}
\end{IEEEeqnarray}
\setcounter{equation}{\value{MYtempeqncnt}}
\end{figure*}

Note that we do not provide the derivation with the simpler delay-dependent description based on the length of the intersection as in previous publications. 
Thus, we present the more realistic case of the area of the intersection ellipse, where the differential scatterers have a two-dimensional displacement. 

\subsection{Hybrid Time Delay Characteristic Probability Density}
\noindent
In this subsection, we derive the hybrid characteristic \ac{pdf} $\rho(t;\xi,\Delta t)$ in delay $\xi$ and time lag $\Delta t$ domains for general \ac{M2M} scattering channels as discussed in Section~\ref{sec:CH_Characterization}. 
Since our starting point is the joint delay Doppler \ac{pdf}, we obtain the hybrid characteristic probability density function by an inverse Fourier transform in the Doppler frequency variable. 
We show that in the limiting case, those newly derived hybrid characteristic \acp{pdf} converge to known results of correlation functions in the literature. 

We obtain the hybrid characteristic \acp{pdf} for the \emph{general case}, as it was defined above, by using the spatial variable $\eta$ instead of the Doppler frequency $\fd$.
By using relationship \eqref{eq:Concept:DelayDoppler}, we obtain with \eqref{eq:doppler_star}
\begin{IEEEeqnarray}{lCl}
\rho(t;\xi,\Delta t)=\sum_{i=1}^2\int p\left(t,q(\xi,\eta,\vartheta);\xi,\fd\right)\mathrm{e}^{\mathrm{j}2\pi \Delta t f_{\mathrm{d},i}^{\star}}\,\mathrm{d}f_{\mathrm{d},i}^{\star}\nonumber\\
=\sum_{i=1}^2\int p\left(t,q(\xi,\eta,\vartheta);\xi,\fd(\eta)\right)\mathrm{e}^{\mathrm{j}2\pi \Delta t f_{\mathrm{d},i}^{\star}(\eta)}\left|\mathbf{J_{s}}\right|\,\mathrm{d}\eta\,.
\end{IEEEeqnarray}
Thus, we can directly insert the weighted spatial scatterer density $p(t;\xi,\eta)$ and perform the inverse Fourier function over $\eta$ as
\begin{IEEEeqnarray}{lCl}
\label{eq:rho_general}
\rho(t;\xi,\Delta t)=\frac{1}{\mathcal{Y}_1}\sum_{i=1}^2\\
\int\limits_{\eta_{1}\left(\xi\right)}^{\eta_{2}\left(\xi\right)}\hspace*{-3pt}\frac{w\left(\xi,\eta\right)l^2\sqrt{A^2+B^2+C^2}\left(\xi^2-\eta^2\right)\mathrm{e}^{\mathrm{j}2\pi \Delta t f_{\mathrm{d},i}^{\star}(t;\xi,\eta)}}{\sqrt{\left(\xi^2-1\right)\left(1-\eta^2\right)\left(A^2+B^2\right)-\left(D-C \xi \eta\right)^2}}\mathrm{d}\eta\,,\nonumber
\end{IEEEeqnarray}
with the time-variant Doppler frequency $f_{\mathrm{d},i}^{\star}(t;\xi,\eta)$ according to \eqref{eq:doppler_star} and the integral limits $\eta_{1}(\xi)$ and $\eta_{2}(\xi)$ according to \eqref{eq:eta12} with $\eta_{2}(\xi)>\eta_{1}(\xi)$.

For the \emph{complementary case} we insert the weighted circular area given in \eqref{eq:area_comp} and can calculate the hybrid characteristic \ac{pdf} by an inverse Fourier transform of the Doppler variable $\fd$ as
\begin{IEEEeqnarray}{lCl}
\rho(t;\xi,\Delta t)=
\frac{1}{\mathcal{Y}_2}\int\limits_{0}^{2\pi}\frac{l^2\left(\xi-\frac{D^2}{C^2\xi^3}\right)}{\left(\xi^2-\left(\frac{D}{C\xi}\right)^2\right)^2}\mathrm{e}^{\mathrm{j}2\pi \fd\left(t;\xi,\vartheta\right)\Delta t}\,\mathrm{d}\vartheta\nonumber\\
=\frac{2l^2\left(\xi-\frac{D^2}{C^2\xi^3}\right)}{\mathcal{Y}_2\left(\xi^2-\left(\frac{D}{C\xi}\right)^2\right)^2}\int\limits_{-f_{\mathrm{l}}}^{f_{\mathrm{l}}}\frac{\mathrm{e}^{\mathrm{j}2\pi\fd \Delta t}}{f_{\mathrm{l}}\sqrt{1-\left(\frac{\fd-f_{\mathrm{o}}}{f_{\mathrm{l}}}\right)^2}}\,\mathrm{d}\fd\\
=\frac{1}{\mathcal{Y}_2}\frac{2\pi l^2\left(\xi-\frac{D^2}{C^2\xi^3}\right)}{\left(\xi^2-\left(\frac{D}{C\xi}\right)^2\right)^2}\mathrm{J}_0\left(2\pi f_{\mathrm{l}}(t;\xi)\Delta t\right)\mathrm{e}^{\mathrm{j}2\pi f_{\mathrm{o}}(t;\xi)\Delta t}\,\nonumber
\end{IEEEeqnarray}
where $\mathrm{J}_0$ is a zeroth-order Bessel function of the first kind, 
\begin{equation}
f_{\mathrm{o}}(t;\xi)= \frac{\fc}{c}\left(\frac{\textstyle{\frac{D}{C}}+1}{\xi+\textstyle\frac{D}{C \xi}}\vtxz+\frac{\textstyle{\frac{D}{C}}-1}{\xi-\textstyle\frac{D}{C \xi}}\vrxz\right)\,,
\end{equation}
is the offset frequency caused by the movement of \ac{TX} and \ac{RX} along the $z$-axis, and 
\begin{IEEEeqnarray}{rCl}
f_{\mathrm{l}}(&t&;\xi)=\frac{\fc}{c}\sqrt{\left(\xi^2-1\right)\left(1-\textstyle\left(\frac{D}{C \xi}\right)^2\right)}\times\\
&&\sqrt{\left(\frac{\vtxx}{\xi+\textstyle\frac{D}{C \xi}}+\frac{\vrxx}{\xi-\textstyle\frac{D}{C \xi}}\right)^2+\left(\frac{\vtxy}{\xi+\textstyle\frac{D}{C \xi}}+\frac{\vrxy}{\xi-\textstyle\frac{D}{C \xi}}\right)^2}\nonumber\,,
\end{IEEEeqnarray}
is the limiting frequency.
The basis for our calculations above is the Doppler frequency in  \cite[(3)]{Walter_TVT21} as 
\begin{IEEEeqnarray}{rCl}
\label{eq:doppler2}
&\fd &(t;\xi,\eta,\vartheta)=\frac{\fc}{c}\Bigg(\\
&&\frac{\xi\eta+1}{\xi+\eta}\vtxz+\frac{\sqrt{\left(\xi^2-1\right)\left(1-\eta^2\right)}}{\xi+\eta}\left(\vtxx\cos\vartheta+\vtxy\sin\vartheta\right)
\nonumber\\
&+&\frac{\xi\eta-1}{\xi-\eta}\vrxz+\frac{\sqrt{\left(\xi^2-1\right)\left(1-\eta^2\right)}}{\xi-\eta}\left(\vrxx\cos\vartheta+\vrxy\sin\vartheta\right)\Bigg)\nonumber
\end{IEEEeqnarray}
where $\vtxvec=\left[\vtxx,\vtxy,\vtxz\right]\tr$ and $\vrxvec=\left[\vrxx,\vrxy,\vrxz\right]\tr$ are the velocity vectors of \ac{TX} and \ac{RX} in the local \ac{CCS}, respectively. 
Since we obtain the relationship $\eta =D/(C\xi)$ for the \emph{complementary case}, the Doppler frequency reduces to $\fd\left(t;\xi,\vartheta\right)$.
Note that in the \emph{complementary case} the delay and Doppler \acp{pdf} factor, and thus are independent of each other.

\subsection{Limiting Value Consideration} 
\noindent
By studying the delay-dependent Doppler \ac{pdf} and the delay-dependent characteristic function in the asymptotic regime, as $\xi\to\infty$, we obtain several expressions that are well-known in the literature.
Specifically, we derive
\begin{IEEEeqnarray}{lll}
\label{eq:Bessel}
\lim_{\xi\to \infty}\rho(t;\Delta t|\xi)=\mathrm{J_0}\left(2\pi \Delta t\frac{\norm{\mathbf{v}_{\mathrm{t}\parallel\mathrm{E}}+\mathbf{v}_{\mathrm{r}\parallel\mathrm{E}}}}{c}\fc\right)\,,\\
\label{eq:Jakes}
\lim_{\xi\to \infty}p(t;\fd|\xi)=\frac{1}{\pi f_{\mathrm{l}\infty}(t)\sqrt{1-\left(\frac{\fd}{f_{\mathrm{l}\infty}(t)}\right)^2}}\,,
\end{IEEEeqnarray}
\begin{IEEEeqnarray}{lll}
\lim_{\xi\to \infty}\mu_{\fd|\xi}(t;\fd)=0\,,\\
\label{eq:Results_xiInf}
\lim_{\xi\to \infty}\sigma_{\fd|\xi}(t;\fd)=\frac{\norm{\mathbf{v}_{\mathrm{t}\parallel\mathrm{E}}+\mathbf{v}_{\mathrm{r}\parallel\mathrm{E}}}}{\sqrt{2}c}\fc=\frac{f_{\mathrm{l}\infty}(t)}{\sqrt{2}}\,,
\end{IEEEeqnarray}
with parallel velocity vectors and limiting Doppler frequency given by
\setcounter{equation}{41}
\begin{IEEEeqnarray}{lll}
\label{eq:vtx_infinity}
\mathbf{v}_{\mathrm{t}\parallel\mathrm{E}}=\frac{\mathbf{n}_{\mathrm{E}}\times\left(\vtxvec\times\mathbf{n}_{\mathrm{E}}\right)}{\norm{\mathbf{n}_{\mathrm{E}}}^2}=\vtxvec-\frac{\left(\vtxvec\cdot\mathbf{n}_{\mathrm{E}}\right)\mathbf{n}_{\mathrm{E}}}{\norm{\mathbf{n}_{\mathrm{E}}}^2}\,,\\
\label{eq:vrx_infinity}
\mathbf{v}_{\mathrm{r}\parallel\mathrm{E}}=\frac{\mathbf{n}_{\mathrm{E}}\times\left(\vrxvec\times\mathbf{n}_{\mathrm{E}}\right)}{\norm{\mathbf{n}_{\mathrm{E}}}^2}=\vrxvec-\frac{\left(\vrxvec\cdot\mathbf{n}_{\mathrm{E}}\right)\mathbf{n}_{\mathrm{E}}}{\norm{\mathbf{n}_{\mathrm{E}}}^2}\,,\\
\label{eq:fd_infinity}
\lim_{\xi\to \infty}f_{\mathrm{l}\infty}(t)=\frac{\norm{\mathbf{v}_{\mathrm{t}\parallel\mathrm{E}}+\mathbf{v}_{\mathrm{r}\parallel\mathrm{E}}}}{c}\fc\,.
\end{IEEEeqnarray}
The velocity vectors $\mathbf{v}_{\mathrm{t}\parallel\mathrm{E}}$ in \eqref{eq:vtx_infinity} and $\mathbf{v}_{\mathrm{r}\parallel\mathrm{E}}$ in \eqref{eq:vrx_infinity} of \ac{TX} and \ac{RX} are parallel to the scattering plane.
The limiting frequencies $f_{\mathrm{l}\infty}(t)$ in \eqref{eq:fd_infinity} for $\xi\to\infty$ are given by the solution of the $\eta$ variable $\pm\left(v_{\mathrm{t}z\parallel\mathrm{E}}+v_{\mathrm{r}z\parallel\mathrm{E}}\right)/\left(\norm{\mathbf{v}_{\mathrm{t}\parallel\mathrm{E}}+\mathbf{v}_{\mathrm{r}\parallel\mathrm{E}}}\right)$ of the polynomial in \cite[eq. (31)]{Walter_TVT21}. 
The result in \eqref{eq:Jakes} matches the classical Jakes result.
The width of the spectrum, however, is determined by the velocity vector components of \ac{TX} and \ac{RX}, which are parallel to the scattering plane.
The reason for this is that for large $\xi$ the eccentricity of the ellipsoid reduces toward $0$, thus, approaching a sphere.
The intersection with the scattering plane thus results in a scattering circle on which the scatterers are uniformly distributed.
The corresponding Fourier transform of the delay-dependent \ac{pdf} in \eqref{eq:Jakes} results in the typical Bessel function in \eqref{eq:Bessel} as the delay-dependent characteristic function $\rho(t;\Delta t|\xi)$.

\section{Comparative Analysis of Theory and Measurement}
\label{sec:results}
\noindent
In our previous works \cite{Walter_TVT17} and \cite{Walter_TWC19}, we primarily examined the scattering contributions in terms of delay and Doppler frequency shift.
However, the decrease in the scattering power with increasing delay, reflecting the influence of the channel's \ac{PDP}, has not yet been addressed and verified.
The typical approach to account for the \ac{PDP} is to empirically adjust the scattering power behavior using a specific path loss exponent, e.g., as in \cite{Walter_PIMRC15}.
In contrast to these empirical approaches, the analytical description of the hybrid characteristic \ac{pdf} of the channel, as presented in this work, allows the calculation of the time-variant \ac{PDP} for any scenario, taking into account the geometry of the environment and the velocities of the transceivers.

In the following, we examine the \ac{A2A} channel as a representative example of a \ac{US} \ac{M2M} channel, being the most general channel where both \ac{TX} and \ac{RX} are not confined to the scattering plane.
We compare data obtained from a measurement campaign \cite{WalterEuCAP10} and the data obtained from the numerical evaluation of the hybrid characteristic \ac{pdf} from \eqref{eq:rho_general} from Section~\ref{sec:HybricFunctions}.

\subsection{Measurement and Simulation Scenario}
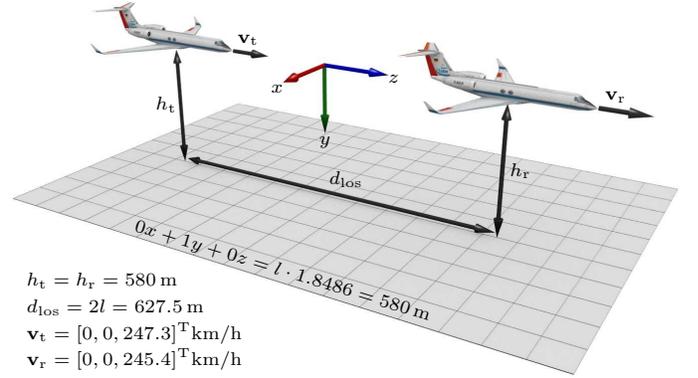
\begin{figure}[tb]
    \centering
    \input{figures/scenario}
	\caption{Aircraft positions, velocity vectors, distance and placement of the scattering plane in a local coordinate system for simulation.}
 \label{fig:scenario}
\end{figure}

\begin{table}[htb]
	\begin{center}
		\caption{Measurement and evaluation parameters}
		\label{tab:chan_param}
            \renewcommand{\arraystretch}{1.3}
		\begin{tabular}{|l|c|c|} \hline
			 \rowcolor{lightgray}\textbf{Parameter Name} & \textbf{Variable} & \textbf{Value}\\
            \hline
		Carrier frequency & $\fc$ & \SI{250}{\mega\hertz} \\  \hline
		Bandwidth & $B$ & \SI{20}{\mega\hertz} \\ \hline
            Frequency resolution & $\Delta f$ & \SI{39.1}{\kilo\hertz} \\ \hline
            Normalized maximum frequency& $\tilde{f}_{\mathrm{max}}$ & $\pm$ 20.92 \\ \hline
            Normalized frequency resolution& $\Delta \tilde{f}$ & 0.082 \\ \hline
            Signal period& $\tau_{\mathrm{max}}$ & \SI{25.6}{\micro\second} \\ \hline
  		Delay resolution & $\Delta \tau$ & \SI{50}{\nano\second} \\ \hline
  		Normalized signal period & $\xi_{\mathrm{max}}$ & 12.24 \\ \hline
  		Normalized delay resolution & $\Delta \xi$ & 0.024 \\ \hline
            Time period &  $t_{\mathrm{max}}$ & \SI{2.1}{\second} \\\hline
		Measurement time grid & $\Delta t$ & \SI{2.048}{\milli\second} \\ \hline
		Max. Doppler frequency & $f_{\mathrm{d,max}}$ & $\pm$\SI{244}{\hertz} \\ \hline
           Doppler resolution &  $\Delta\fd$ & \SI{0.5}{\hertz} \\ \hline 
		\end{tabular}
	\end{center}
\end{table}

\noindent
We consider a scenario in which two aircraft fly at the same altitude above ground and are positioned $d_{\mathrm{los}}=\tau_{\mathrm{los}}c$ behind each other.
An overview of the scenario including positions of the aircraft, velocity vectors $\mathbf{v}_{\mathrm{t}}$ and $\mathbf{v}_{\mathrm{r}}$ in a local \ac{CCS}, and the parameters of the scattering plane for the simulation is shown in Fig.~\ref{fig:scenario}.

The measurement parameters that are given by the channel sounding equipment are provided in Tab.~\ref{tab:chan_param}. 
The normalized parameters $\xi$ and $\tilde{f}$ can be derived from the physical parameters and the line-of-sight delay $\tau_{\mathrm{los}}$.

In order to show the advantages of the probability based channel description functions with respect to real world measurements and the corresponding correlation functions, we first compare normalized versions of both probability based and correlation based functions. 
In a second step, we use the non-normalized probabilistic 2D functions and marginalize them to obtain the delay and Doppler spectra and the correlation in the time and frequency domain.
We discuss the 2D functions in the same order as in Section~\ref{sec:Prob_based}.

\subsection{Normalized Probability and Correlation Based Functions}
\noindent
For continuity with our previous paper \cite{Walter_TVT21}, we begin by examining the factorized \ac{pdf} $p(t;\fd|\xi)=p(t;\xi,\fd)/p(t;\xi)$ in Fig.~\ref{fig:LSF_theo} and compare it with the real part of the time-variant, delay-dependent \ac{LSF} $\tilde{\mathcal{C}}_{\mathbf{H}}(t;\fd|\xi)$ in Fig.~\ref{fig:LSF}.
Additionally, we illustrate the delay-dependent mean Doppler $\mu_{\fd|\xi}(t)$ from \eqref{eq:mu_fd} and Doppler spread $\sigma_{\fd|\xi}(t)$ from \eqref{eq:sigma_fd} in Fig.~\ref{fig:LSF_theo}.
Since the spectra are symmetric, the mean Doppler stays zero, but the Doppler spread is increasing with delay $\xi$ and approaches, according to \eqref{eq:Results_xiInf}, the value $\sigma_{\fd|\xi\to\infty}(t)=\SI{80.65}{\hertz}$.
In the limiting case, the delay-dependent spectrum conforms to a Jakes spectrum according to \eqref{eq:Jakes}, consistent with the literature.
The analysis of the theoretical results in Fig.~\ref{fig:LSF_theo} reveals that the shape and the values of both the probability based and the correlation based functions are the same. 
The scattering power in the measurement data is very weak and close to the noise threshold.
Furthermore, the scattering does not occur uniformly on the ground as in our assumption.
Thus, we can observe gaps in Fig.~\ref{fig:LSF}.
This will lead to slight differences in the marginalized one-dimensional functions. 
The structure of the measured channel, however, is well captured with the delay-dependent Doppler pdf $p(t;\fd|\xi)$. 

\begin{figure}[t]
	\centering  \includegraphics[width=0.9\columnwidth]{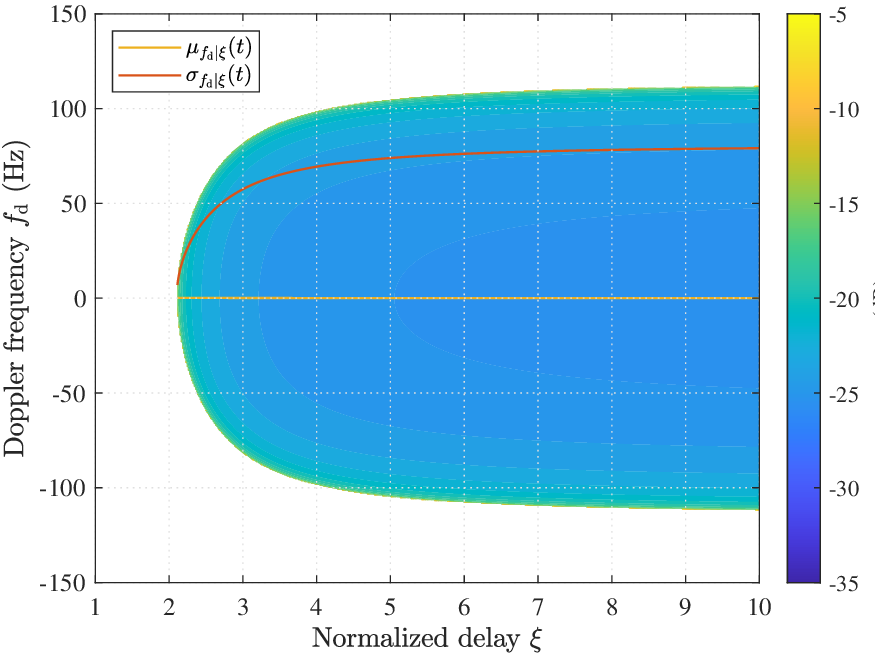}
	\caption{Theoretical time-variant, delay-dependent Doppler pdf $p(t;\fd|\xi)$ with mean Doppler $\mu_{\fd|\xi}(t)$ and Doppler spread $\sigma_{\fd|\xi}(t)$.}
	\label{fig:LSF_theo}
\end{figure}

\begin{figure}[t]
	\centering  \includegraphics[width=0.9\columnwidth]{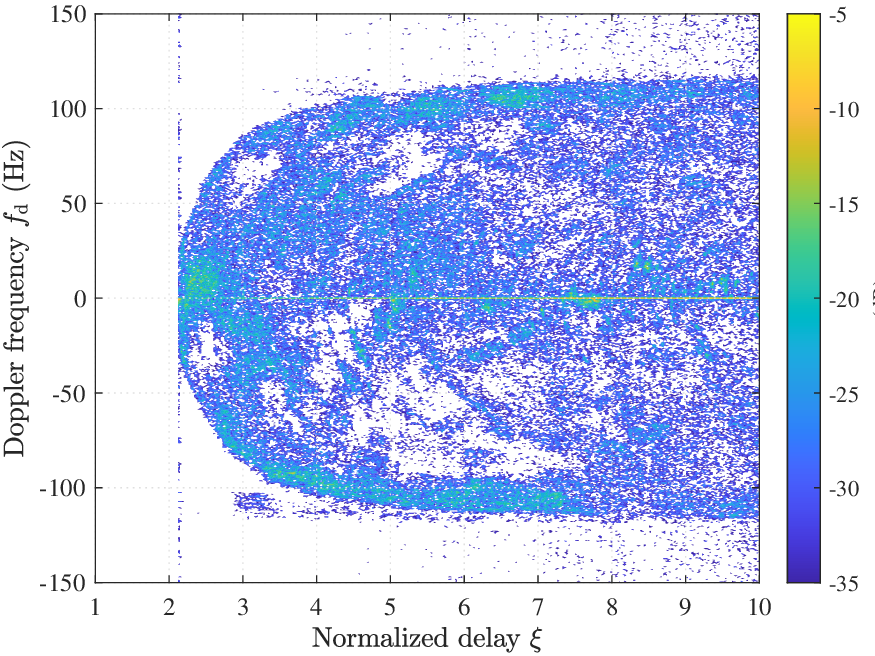}
	\caption{Measured time-variant, delay-dependent local scattering function $\Re\left\{\tilde{\mathcal{C}}_{\mathbf{H}}(t;\fd|\xi)\right\}$.}
	\label{fig:LSF}
\end{figure}

We continue our comparison with the real part of the delay-dependent hybrid time delay characteristic \ac{pdf} $\rho(t;\Delta t|\xi)$ in Fig.~\ref{fig:hybrid_char_pdf}. 
The real part of the normalized temporally correlated delay-dependent delay cross power density $P_h(t;\Delta t|\xi)$ is given in Fig.~\ref{fig:cross_power}.
Both theoretical results and measurement data demonstrate a strong agreement in the temporal correlation of the channel.
The correlation decreases noticeably with increasing delay.
For large delays, $\xi\to\infty$, the delay-dependent characteristic function converges to a Bessel function as described in \eqref{eq:Bessel} aligning well with theoretical expectations.
Since both the delay-dependent Doppler \ac{pdf} and delay-dependent, time-variant \ac{LSF} get wider with increasing delay, the correlation in the $\Delta t$ variable naturally decreases, as can be seen in Figs. \ref{fig:LSF_theo} and \ref{fig:LSF}.
The influence of the \ac{LOS} signal and the \ac{SR} reflection, which was eliminated from the measurement data for comparison reasons, is still slightly observable for delays $\xi$ close to $\xi_{\mathrm{sr}}=2.1$ according to \eqref{eq:xi_sr2}.

\begin{figure}[t]
	\centering    \includegraphics[width=0.9\columnwidth]{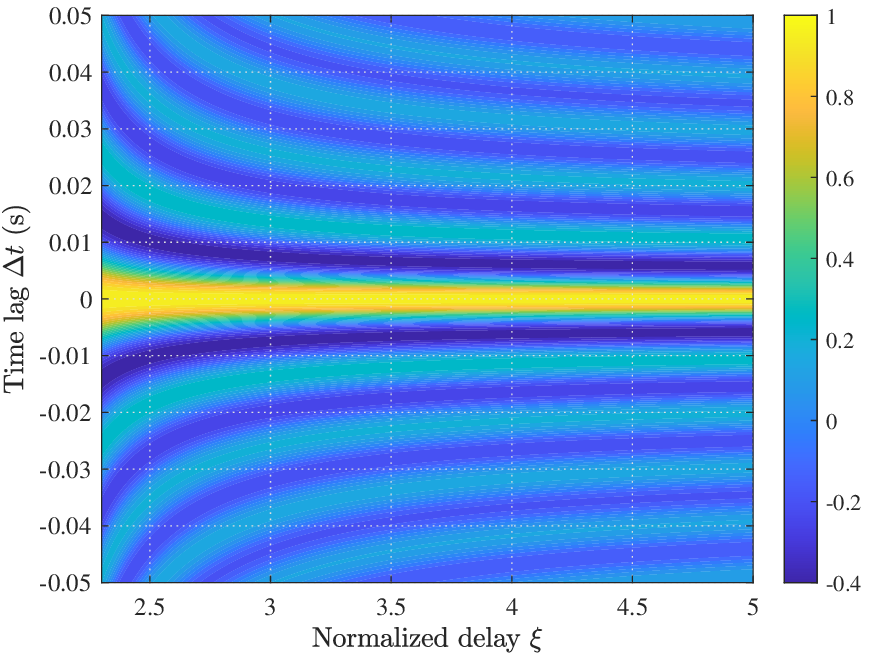}
	\caption{Theoretical time-variant, delay-dependent temporal characteristic function $\Re\left\{\rho(t;\Delta t|\xi)\right\}$.}	\label{fig:hybrid_char_pdf}
\end{figure}

\begin{figure}[t]
	\centering    \includegraphics[width=0.9\columnwidth]{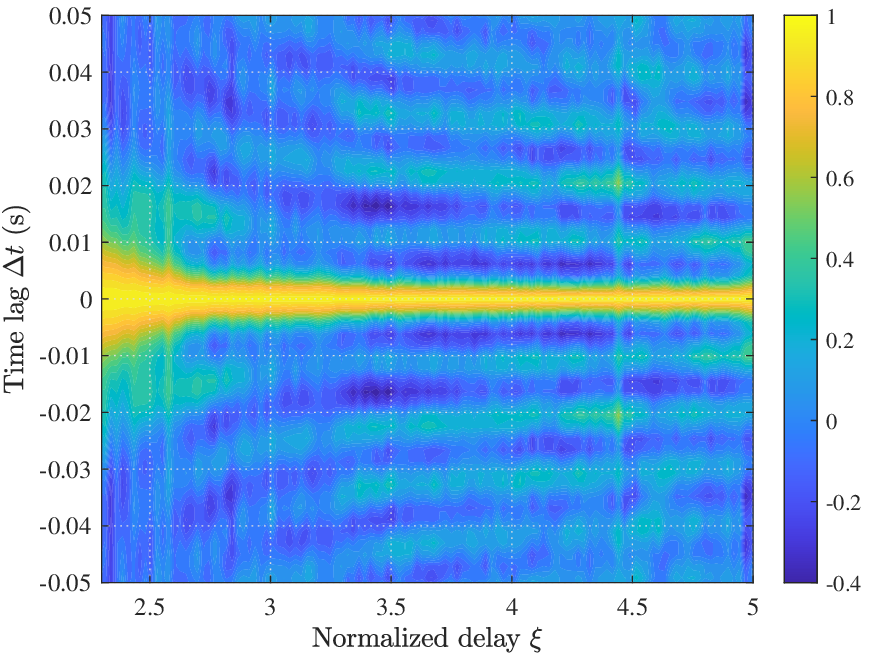}
	\caption{Measured time-variant, delay-dependent temporal correlation function $\Re\left\{P_h(t;\Delta t|\xi)\right\}$.}
	\label{fig:cross_power}
\end{figure}

Next, we compare the real parts of the newly introduced time-variant hybrid Doppler-dependent frequency characteristic \ac{pdf} $\varrho(t; \Delta \tilde{f}|\fd)$ with $R_{\varrho}(t; \Delta \tilde{f}|\fd)$.
The time-variant hybrid frequency Doppler characteristic \ac{pdf} $\varrho(t; \Delta \tilde{f},\fd)$ is thus divided by its Doppler spectral density $p(t;\fd)$.
Therefore the correlation in normalized frequency lag $\Delta\tilde{f}$ becomes visible.
The theoretical results in Fig.~\ref{fig:dop_freq_corr_theo} show that the correlation is largest for the Doppler frequency $\fd=\SI{0}{\hertz}$.
Both with increasing and decreasing Doppler frequency, the correlation symmetrically diminishes along the frequency axis.
The correlation of the measurement data in Fig.~\ref{fig:dop_freq_corr} shows a similar behavior.  

\begin{figure}[t]
	\centering    \includegraphics[width=0.9\columnwidth]{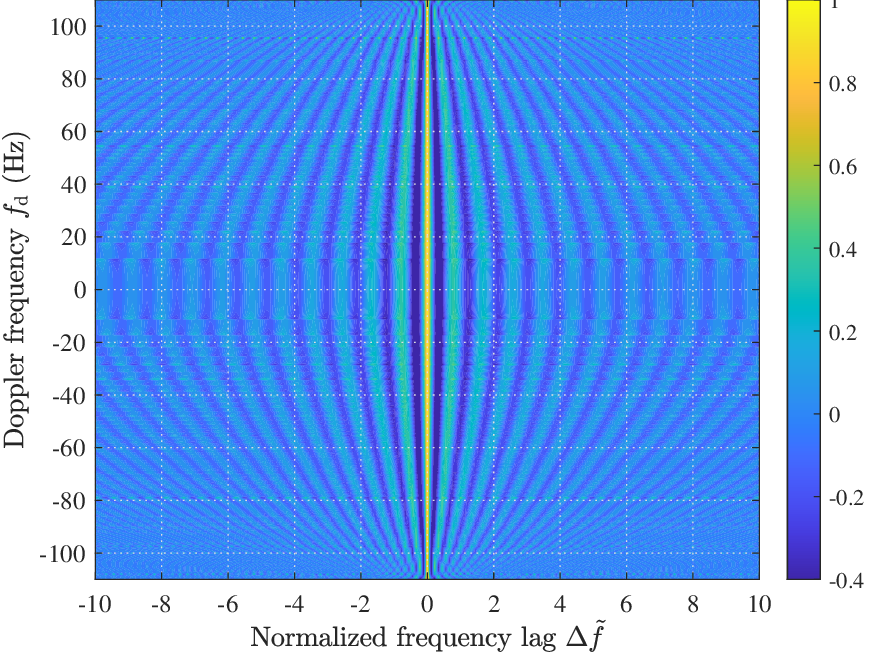}
	\caption{Theoretical time-variant Doppler-dependent frequency characteristic function $\Re\left\{\varrho(t; \Delta \tilde{f}|\fd)\right\}$.}
	\label{fig:dop_freq_corr_theo}
\end{figure}

\begin{figure}[t]
	\centering  \includegraphics[width=0.9\columnwidth]{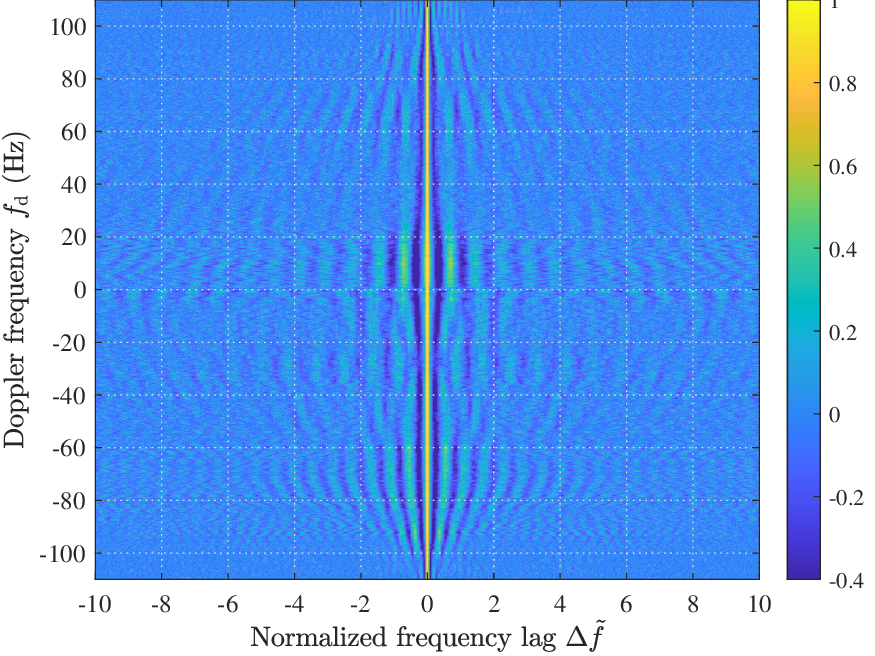}
	\caption{Measured time-variant Doppler-dependent frequency correlation function $\Re\left\{R_{\varrho}(t; \Delta \tilde{f}|\fd)\right\}$.}
	\label{fig:dop_freq_corr}
\end{figure}

Finally, Fourier transforms of both hybrid characteristic \acp{pdf} lead to the joint time-variant characteristic function.
The real part of the time-variant joint time frequency characteristic function $r(t;\Delta \tilde{f},\Delta t)$ and the real part of the time frequency correlation function $\tilde{R}_L(t; \Delta \tilde{f},\Delta t)$ are shown in Fig.~\ref{fig:time_freq_corr_theo} and Fig.~\ref{fig:time_freq_corr}.
They both have a peak at zero time and zero frequency shift.
Along the time and frequency axes, both functions further exhibit the typical decreasing correlation behavior.

\begin{figure}[t]
	\centering    \includegraphics[width=0.9\columnwidth]{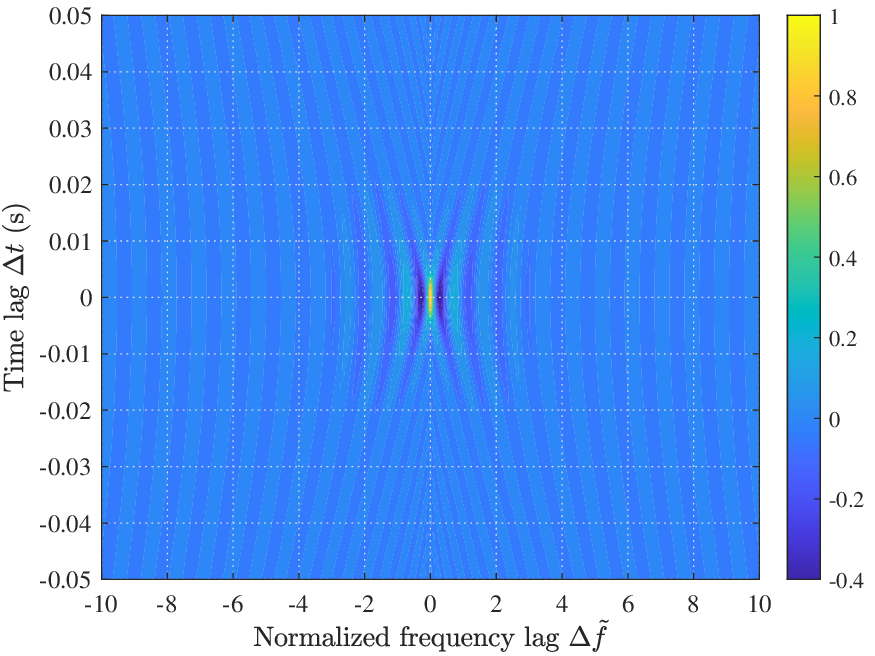}
	\caption{Theoretical time-variant, joint time frequency characteristic function $\Re\left\{r(t; \Delta \tilde{f},\Delta t)\right\}$.}
	\label{fig:time_freq_corr_theo}
\end{figure}

\begin{figure}[t]
	\centering  \includegraphics[width=0.9\columnwidth]{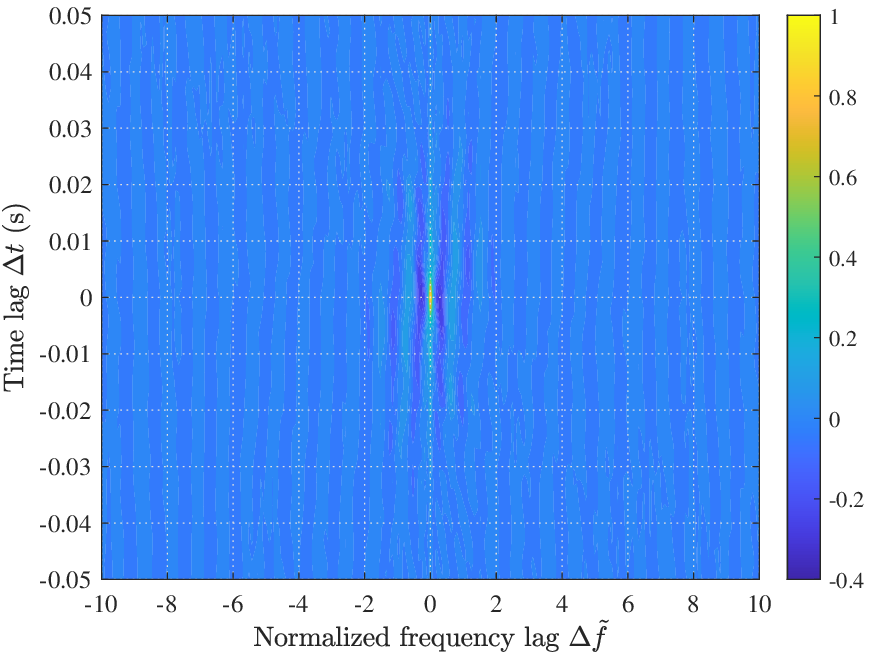}
	\caption{Measured time-variant, joint time frequency correlation function $\Re\left\{\tilde{R}_L(t; \Delta \tilde{f},\Delta t)\right\}$.}
	\label{fig:time_freq_corr}
\end{figure}

\subsection{Time-Variant Probability Based Functions}
\noindent
In this subsection, we examine the full time-variant 2D probabilistic description of the channel.
In reference to Fig.~\ref{fig:correlation_functions}, the block diagram in Fig.~\ref{fig:SampledEvaluations} shows the relationship between the different channel descriptions and their corresponding mutual transformations.

We begin with the joint delay Doppler \ac{pdf} $p(t;\xi,\fd)$ in Fig.~\ref{fig:SampledEvaluations}(I). 
Note the similarity to the conditional \ac{pdf} in Fig.~\ref{fig:LSF_theo}.
Yet, the key difference is the obvious drop of the probability mass, or equivalently signal power, with increasing delay. 
This drop of probability with increasing delay, as captured by the marginal $p(t; \xi)$ in Fig.~\ref{fig:SampledEvaluations}(i), is explicitly transferred to the hybrid time delay characteristic \ac{pdf} $\rho(t;\xi, \Delta t)$, see Fig.~\ref{fig:SampledEvaluations}(II).
Notably, the joint \ac{pdf} descriptions correctly account for the weighting of the functions in the delay domain.
A comparable weighting occurs in the Doppler domain with the Doppler \ac{pdf} $p(t; \fd)$ affecting $\varrho(t;\Delta \tilde{f}, \fd)$ as shown in Fig.~\ref{fig:SampledEvaluations}(III).
Naturally, the joint characteristic function $r(t;\Delta\tilde{f},\Delta t)$ in Fig.~\ref{fig:SampledEvaluations}(IV) accounts for both of these weightings implicitly through the Fourier transform. 

The clear advantage of these joint descriptions is their computability from environmental models and location information of transceivers, as demonstrated in \cite{Bellido_TVT22}.
Moreover, they enable the calculation of four time-varying marginalized descriptions: the delay \ac{pdf} $p(t;\xi)$, the Doppler \ac{pdf} $p(t;\fd)$, as well as the temporal characteristic function $r(t;\Delta t)$ and the frequency characteristic function $r(t;\Delta \tilde{f})$ following the properties of \acp{pdf} or corresponding characteristic functions.
In Fig.~\ref{fig:SampledEvaluations}, the respective relationships of these marginalized descriptions are depicted by arrows.

The time-variant probability densities $p(t;\xi)$ and $p(t;\fd)$ are computed by integrating the joint delay Doppler \ac{pdf} $p(t;\xi,\fd)$ of Fig.~\ref{fig:SampledEvaluations}(I), where the integration variable is $\fd$ for $p(t;\xi)$ and $\xi$ for $p(t;\fd)$. 
Alternatively, the same result can be obtained by setting the Delta variables $\Delta t=0$ in $\rho(t;\xi,\Delta t)$, Fig.~\ref{fig:SampledEvaluations}(II) and $\Delta \tilde{f}=0$ in $\varrho(t;\Delta\tilde{f},\fd)$, Fig.~\ref{fig:SampledEvaluations}(III), respectively.
This is a general property of a characteristic function.  
As illustrated in Figs.~\ref{fig:SampledEvaluations}(i)-(ii), the analytically computed \acp{pdf} $p(t;\xi)$ and $p(t;\fd)$ align remarkably well with those derived from measurement data.
Note again, as mentioned above, that in order to reveal the scattering in the channel, the signal components from \ac{LOS} and \ac{SR} are eliminated.
Therefore, the decreasing behavior of the scatter channel caused by path loss can be clearly observed in Fig.~\ref{fig:SampledEvaluations}(i) and is well reflected by the delay \ac{pdf}. 
Regarding the Doppler \ac{pdf} shown in Fig.~\ref{fig:SampledEvaluations}(ii), it is evident that the shape deviates from the traditional Jakes spectrum.  
While the theoretical model shows a concave shape, the measurements exhibit a higher probability at zero Doppler due to the imperfect elimination of \ac{LOS} and \ac{SR}, along with a slight increase at the limiting Doppler frequencies.

Finally, let us analyze the two time-variant characteristic functions $r(t; \Delta\tilde{f})$ and $r(t;\Delta t)$.
The former, shown in Fig.~\ref{fig:SampledEvaluations}(iii), can be obtained via marginalization of the hybrid frequency Doppler characteristic \ac{pdf} $\varrho(t;\Delta \tilde{f},\fd)$ or by setting $\Delta t=0$ in $r(t;\Delta \tilde{f},\Delta t)$.
The temporal characteristic function, shown in Fig.~\ref{fig:SampledEvaluations}(iv), $r(t;\Delta t)$ can be similarly computed from $\rho(t;\xi,\Delta t)$ or $r(t;\Delta \tilde{f},\Delta t)$.
We observe that the zero crossings and sidelobes of the theoretical curves closely match those computed from measurement data both in Fig.~\ref{fig:SampledEvaluations}(iii) and Fig.~\ref{fig:SampledEvaluations}(iv).
Further, we can determine the coherence bandwidth as a solution to $\Re\left\{r(t; \Delta \tilde{f}=0,\Delta t)\right\}=1/2$.
The normalized coherence bandwidth is about $B_{\mathrm{C}}=0.126$, which corresponds to a physical bandwidth of $B_{\mathrm{C'}}=\SI{60.239}{\kilo\hertz}$.
Note that in Fig.~\ref{fig:SampledEvaluations}(iv) the empirical evaluations show a slight elevation of the sidelobes.
This discrepancy can again be attributed to the imperfect elimination of the \ac{LOS} and \ac{SR} components.
Equivalently to the coherence bandwidth, we derive the channel coherence time as a solution to $\Re\left\{r(t; \Delta \tilde{f},\Delta t=0)\right\}=1/2$, resulting in $T_{\mathrm{C}}=\SI{6.4}{\milli\second}$.

\begin{figure*}    
    \resizebox{\linewidth}{!}{\input{figures/correlation_functions_results}}
    \caption{Time-variant probability based functions from Fig.~\ref{fig:correlation_functions} for an aircraft-to-aircraft scenario with plane parameters $A=0$, $B=1$, $C=0$, $D=1.8486$, and $l=\SI{313.75}{\meter}$, with velocity vectors of the transmitter $\mathbf{v}_{\mathrm{t}}={[0,0,247.3]}\tr\si{\kilo\meter/\hour}$ and the receiver $\mathbf{v}_{\mathrm{r}}={[0,0,245.4]}\tr\si{\kilo\meter/\hour}$ according to Fig.~\ref{fig:scenario}.}
    \label{fig:SampledEvaluations}
\end{figure*}
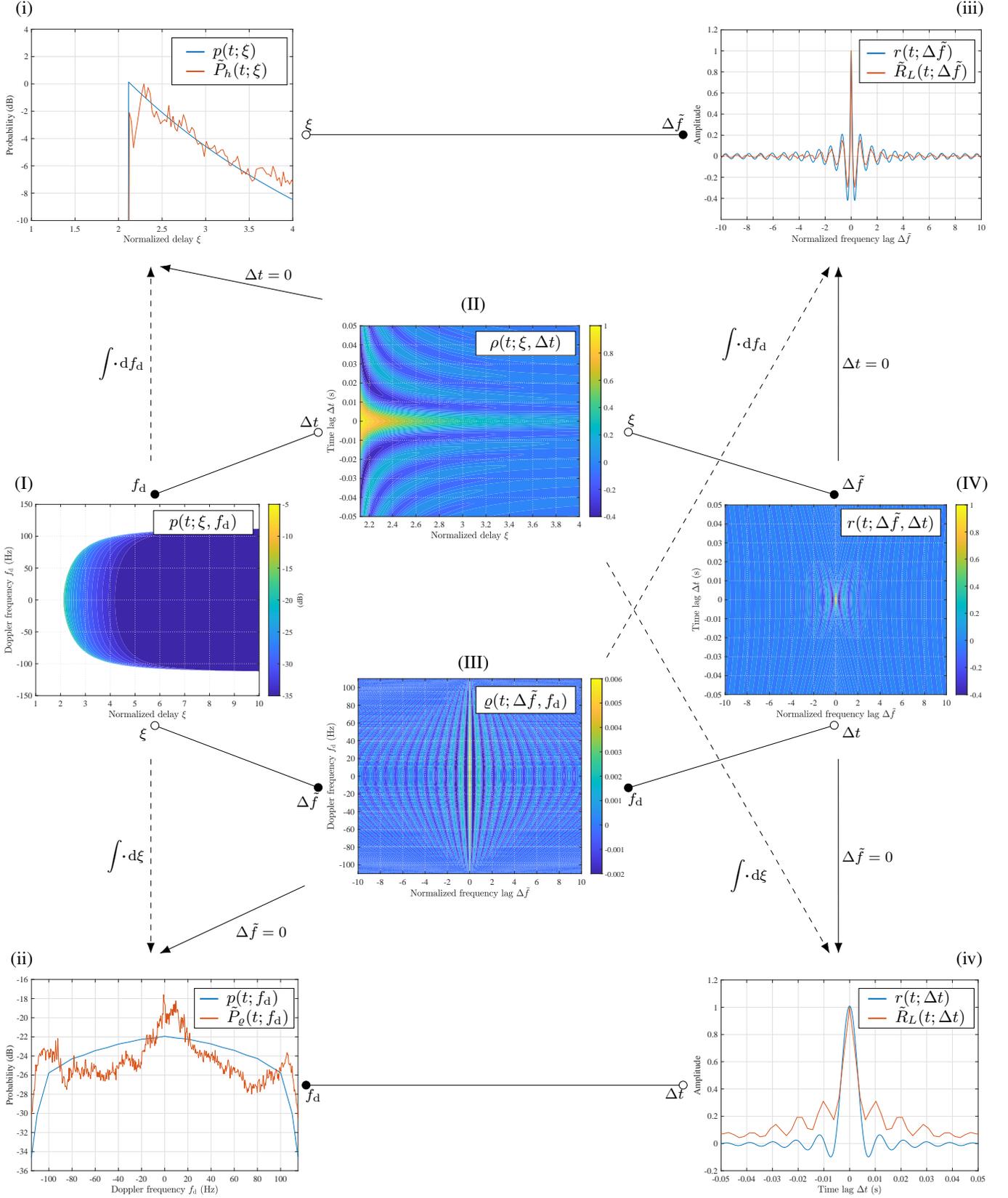

\section{Conclusion}
\noindent
We have presented a complete analytic probability based description of the mobile-to-mobile uncorrelated scatter channel.
The probability based description is proportional to the correlation based description introduced by Bello for wide-sense stationary, uncorrelated scattering channels and by Matz for the general case.

A new set of functions, which we term the hybrid characteristic probability density function is introduced and derived. 
These functions have hybrid properties in the sense that in one variable they behave like a probability density function, while in the other they act as a characteristic function.
The hybrid time delay characteristic probability density functions are shown to be proportional to the temporally correlated delay cross-power spectral density as introduced by Bello.
The complete two-dimensional description allows for the scattering-based path loss to be naturally included in the model. 

The verification of the probability based description is done by using the measurement data from an air-to-air measurement campaign.
This scenario is the most general mobile-to-mobile channel, since the transmitter and receiver are arbitrarily located in 3D space. 
Through appropriate normalization of both probability based and correlation based functions, the two approaches can be directly compared without determining the proportionality constant.
The comparison shows that the proposed new probabilistic description closely aligns with the measurements.

\section*{Acknowledgment}
\noindent
Special thanks go to Rohde $\&$ Schwarz GmbH $\&$ Co. KG. for providing the air-to-air measurement data.

\bibliographystyle{IEEEtran}
\bibliography{IEEEabrv,literature}

\end{document}

%% file: figures/tikz_def.tex
\usepackage{mathrsfs,amsmath}
\usepackage{circuitikz}
\usetikzlibrary{positioning,fit,calc} 
\usetikzlibrary[arrows]

\def\blockWitdh{0.9in}
\def\blockHeight{0.45in}
\def\vertDist{0.83in}
\def\columnWitdh{3.45in}
\def\trafoSep{5pt}

\tikzset{block/.style={draw, thick, font = \footnotesize, text width=\blockWitdh, minimum height=\blockHeight, align=center},   
    line/.style={-latex}, inner sep= 0,outer sep = 2pt}

%% file: figures/us_correlation_functions_v3.tex
\begin{tikzpicture}  

\node[block] (Ps) at (\blockWitdh/2,0.22in) {$p(t; \xi,\fd)$};
\node (mid) at (\columnWitdh/2,0.22in) {};
\node[block] (rM) at (\columnWitdh-\blockWitdh/2,0.22in) {$r(t; \Delta \tilde{f},\Delta t)$};
\node[block] (Ph) at ($(mid)+(0,\vertDist)$) {$\rho(t; \xi,\Delta t)$};
\node[block] (P1) at ($(mid)+(0,-\vertDist)$) {$\varrho(t; \Delta \tilde{f},\fd)$};

\node (tBlock) at (mid) [draw,dashed,minimum width=\columnwidth+0.2in,minimum height=\blockHeight+2*\vertDist+0.2in] {};

\node[block] (P2) at ($(P1)+(0,-\vertDist)$) {$\varrho(\Delta\fd; \xi, \Delta t)$};
\node[block] (PV) at ($(Ps)+(0,-3*\vertDist)$) {$\rho(\Delta\fd;\xi,\fd)$};
\node[block] (rf) at ($(rM)+(0,-3*\vertDist)$) {$\mathcal{R}(\Delta\fd; \Delta \tilde{f}, \Delta t)$};
\node[block] (rH) at ($(P2)+(0,-2*\vertDist)$) {$r(\Delta\fd; \Delta \tilde{f},\fd)$};

\draw[*-o] (rH.west) -- (PV.south);
\draw[*-o] (rH.east) -- (rf.south);
\draw[*-o] (Ps.north) -- (Ph.west);
\draw[*-o] (rM.north) -- (Ph.east);

\draw[*-o] ($(PV.north)+(-.3in,0)$) -- ($(Ps.south)+(-.3in,0)$);
\draw[*-o] ($(rf.north)+(.3in,0)$) -- ($(rM.south)+(.3in,0)$);

\draw[*-o] (P1.west) -- (Ps.south);
\draw[*-o] (P1.east) -- (rM.south);
\draw[*-o] (PV.north) -- (P2.west);
\draw[*-o] (rf.north) -- (P2.east);

\node[font=\footnotesize, below left] at (PV.south) {$\xi$};
\node[font=\footnotesize, above left] at (P2.west) {$\Delta t$};
\node[font=\footnotesize, above left] at (PV.north) {$\fd$};
\node[font=\footnotesize, below left] at (rH.west) {$\Delta \tilde{f}$};
\node[font=\footnotesize, above left] at ($(PV.north)+(-.32in,0)$) {$\Delta\fd$};
\node[font=\footnotesize, below left] at ($(Ps.south)+(-.32in,0)$) {$t$};
\node[font=\footnotesize, above right] at (P2.east) {$\xi$};
\node[font=\footnotesize, below right] at (rf.south) {$\Delta t$};
\node[font=\footnotesize, above right] at ($(rf.north)+(.32in,0)$) {$\Delta\fd$};
\node[font=\footnotesize, below right] at ($(rM.south)+(.32in,0)$) {$t$};
\node[font=\footnotesize, above right] at (rM.north) {$\Delta \tilde{f}$};
\node[font=\footnotesize, below right] at (P1.east) {$\fd$};
\node[font=\footnotesize, below right] at (rM.south) {$\Delta t$};
\node[font=\footnotesize, above right] at (Ph.east) {$\xi$};
\node[font=\footnotesize, above left] at (Ph.west) {$\Delta t$};
\node[font=\footnotesize, above left] at (Ps.north) {$\fd$};
\node[font=\footnotesize, below left] at (Ps.south) {$\xi$};
\node[font=\footnotesize, below left] at (P1.west) {$\Delta \tilde{f}$};
\node[font=\footnotesize, above right] at (rf.north) {$\Delta \tilde{f}$};
\node[font=\footnotesize, below right] at (rH.east) {$\fd$};

\end{tikzpicture}  

%% file: figures/correlation_functions_v2.tex
\begin{tikzpicture}  
\node (center) at (0,0) {};
\node[block] (STT) at ($(center)+(0,-\vertDist)$) {$\varrho(t; \Delta \tilde{f},\fd)$};
\node[block] (Sss) at ($(center)+(-\columnWitdh/2+\blockWitdh/2,0)$) {$p(t; \xi,\fd)$};
\node[block] (SHH) at ($(center)+(\columnWitdh/2-\blockWitdh/2,0)$) {$r(t; \Delta \tilde{f}, \Delta t)$};
\node[block] (Shh) at ($(center)+(0,\vertDist)$) {$\rho(t; \xi, \Delta t)$};

\node[block] (Pxi) at ($(Sss)+(0,2*\vertDist+0.22in)$) {$p(t; \xi)$};
\node[block] (Pfd) at ($(Sss)+(0,-2*\vertDist-0.22in)$) {$p(t; \fd)$};
\node[block] (Stt) at ($(SHH)+(0,-2*\vertDist-0.22in)$) {$r(t; \Delta t)$};
\node[block] (Sff) at ($(SHH)+(0,2*\vertDist+0.22in)$) {$r(t; \Delta \tilde{f})$};

\draw[*-o] (STT.west) -- (Sss.south);
\draw[*-o] (STT.east) -- (SHH.south);
\draw[*-o] (SHH.north) -- (Shh.east);
\draw[*-o] (Sss.north) -- (Shh.west);

\draw[*-o] (Sff.west) -- (Pxi.east);
\draw[*-o] (Pfd.east) -- (Stt.west);

\node[font=\footnotesize, below left] at (STT.west) {$\Delta \tilde{f}$};
\node[font=\footnotesize, below left] at (Sss.south) {$\xi$};
\node[font=\footnotesize, below right] at (STT.east) {$\fd$};
\node[font=\footnotesize, below right] at (SHH.south) {$\Delta t$};

\node[font=\footnotesize, above left] at (Sss.north) {$\fd$};
\node[font=\footnotesize, above left] at (Shh.west) {$\Delta t$};
\node[font=\footnotesize, above right] at (SHH.north) {$\Delta \tilde{f}$};
\node[font=\footnotesize, above right] at (Shh.east) {$\xi$};
\node[font=\footnotesize, above left] at (Sff.west) {$\Delta \tilde{f}$};
\node[font=\footnotesize, above right] at (Pxi.east) {$\xi$};
\node[font=\footnotesize, below right] at (Pfd.east) {$\fd$};
\node[font=\footnotesize, below left] at (Stt.west) {$\Delta t$};

\draw[-{Latex[scale=1.2]},dashed] ($(Sss.south) + (0,-2*\trafoSep)$) -- ($(Pfd.north) + (0,\trafoSep)$)node [midway,fill=white] {\footnotesize $\displaystyle\int \hspace{-2.0pt} \boldsymbol{\cdot} \hspace{1.0pt} \mathrm{d}\xi\,$};
\draw[-{Latex[scale=1.2]}] ($(STT.south) + (0,-\trafoSep)$) -- ($(Pfd.north) + (0.5*\trafoSep,\trafoSep)$)node [midway,below right] {\footnotesize $\Delta \tilde{f}=0$};
\draw[-{Latex[scale=1.2]}] ($(SHH.south) + (0,-2*\trafoSep)$) -- ($(Stt.north) + (0,\trafoSep)$)node [midway,right] {\footnotesize $\Delta \tilde{f}=0$};
\draw[-{Latex[scale=1.2]}] ($(SHH.north) + (0,2*\trafoSep)$) -- ($(Sff.south) + (0,-\trafoSep)$)node [midway,right] {\footnotesize $\Delta t=0$};
\draw[-{Latex[scale=1.2]}] ($(Shh.north) + (0,\trafoSep)$) -- ($(Pxi.south) + (0.5*\trafoSep,-\trafoSep)$)node [midway,above right] {\footnotesize $\Delta t=0$};
\draw[-{Latex[scale=1.2]},dashed] ($(Sss.north) + (0,2*\trafoSep)$) -- ($(Pxi.south) + (0,-\trafoSep)$)node [midway,fill=white] {\footnotesize $\displaystyle\int \hspace{-2.0pt} \boldsymbol{\cdot} \hspace{1.0pt} \mathrm{d}\fd\,$};

\draw[-{Latex[scale=1.2]},dashed] ($(STT.north) + (0,\trafoSep)$) -- ($(Sff.south) + (-0.5*\trafoSep,-\trafoSep)$)node [near end, fill=white] {\footnotesize $\displaystyle\int \hspace{-2.0pt} \boldsymbol{\cdot} \hspace{1.0pt} \mathrm{d}\fd\,$};
\draw[-{Latex[scale=1.2]},dashed] ($(Shh.south) + (0,-\trafoSep)$) -- ($(Stt.north) + (-0.5*\trafoSep,\trafoSep)$)node [near end,fill=white] {\footnotesize $\displaystyle\int \hspace{-2.0pt} \boldsymbol{\cdot} \hspace{1.0pt} \mathrm{d}\xi\,$};
\end{tikzpicture}  

%% file: figures/scenario.tex
\begin{tikzpicture}
    \node (origin) at (0,0) {\includegraphics[width=\columnwidth]{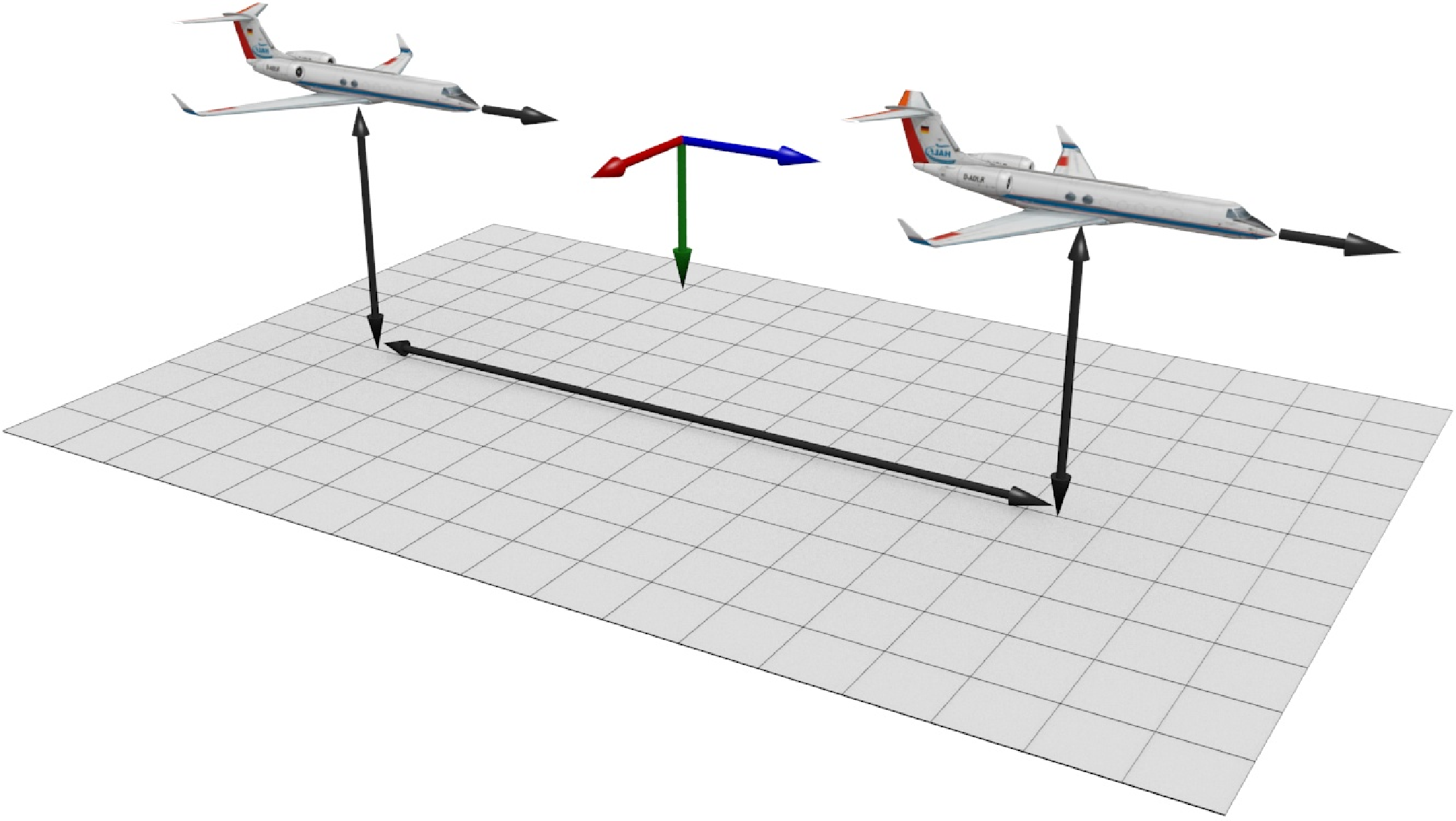}};
    \node (x) at (-26pt,38pt) {\scriptsize{$x$}};
    \node (y) at (-8pt,18pt) {\scriptsize{$y$}};
    \node (z) at (18pt,42pt) {\scriptsize{$z$}};
    \node (vTx) at (-37pt,56.5pt) {\scriptsize{$\mathbf{v}_{\mathrm{t}}$}};
    \node (vRx) at (102pt,35pt) {\scriptsize{$\mathbf{v}_{\mathrm{r}}$}};
    \node (hTx) at (-68pt,32pt) {\scriptsize{$\htx$}};
    \node (hRx) at (66pt,7pt) {\scriptsize{$\hrx$}};
    \node (dist) at (0pt,4pt) {\scriptsize{$d_{\mathrm{los}}$}};
    \node[rotate=-17] (plane) at (-24pt,-31pt) {\scriptsize{$0x+1y+0z=l\cdot1.8486=\SI{580}{\meter}$}};
    \node at (-80pt,-50pt) [draw=none, thin, inner sep=1pt, fill=none,minimum width=10pt,minimum height=10pt,align=left,font=\small] 
    {
    \scriptsize{$\htx = \hrx = \SI{580}{\meter}$}\\
    \scriptsize{$d_{\mathrm{los}}=2l=\SI{627.5}{\meter}$}\\
    \scriptsize{$\mathbf{v}_{\mathrm{t}}={[0,0,247.3]}\tr\si{\kilo\meter/\hour}$} \\
    \scriptsize{$\mathbf{v}_{\mathrm{r}}={[0,0,245.4]}\tr\si{\kilo\meter/\hour}$}
    };
\end{tikzpicture}  

%% file: figures/correlation_functions_results.tex
\begin{tikzpicture}[scale=1.9]
\definecolor{Matlabblue}{rgb}{0, 0.45, 0.74}
\definecolor{Matlaborange}{rgb}{0.85, 0.33, 0.1}
\DeclareRobustCommand\mytikzLineMatlabBlue{\tikz[baseline=-0.5ex]{\draw[color=Matlabblue,  thick] (0,0) -- (0.35cm,0);}}
\DeclareRobustCommand\mytikzLineMatlabOrange{\tikz[baseline=-0.5ex]{\draw[color=Matlaborange,  thick] (0,0) -- (0.35cm,0);}}

\node (center) at (0,0) {};

\node[label=above:{(III)}] (STT) at (0,-1.8) {\includegraphics[width=0.31\textwidth]{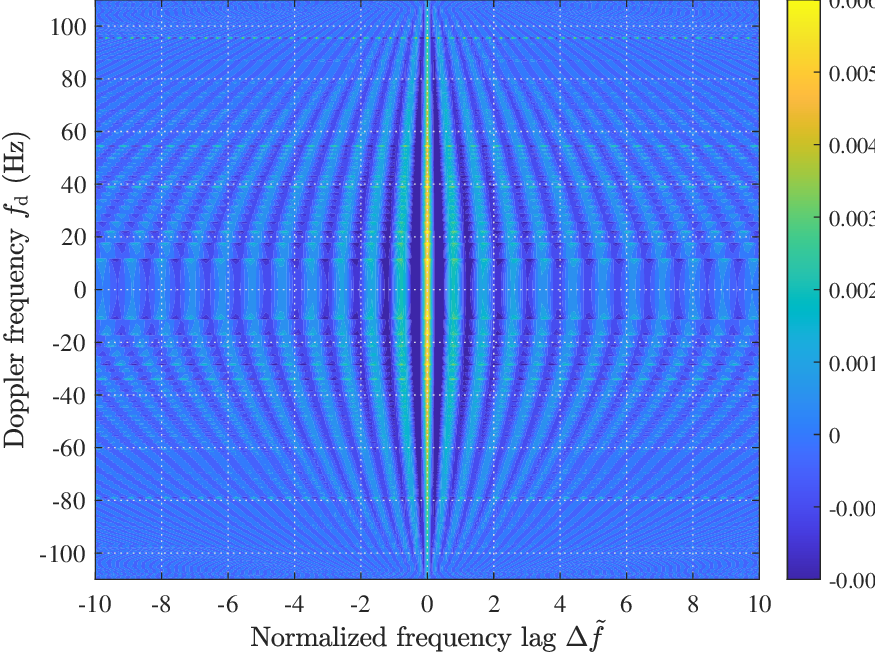}};
\node at ($(STT)+(15pt,25.5pt)$) [draw=black, thin, inner sep=2pt, fill=white,minimum width=1.9cm,minimum height=0.9,align=left,font=\small] 
{$\varrho(t; \Delta \tilde{f}, \fd)$};  

\node[label=above:{(II)}] (Shh) at (0,+1.8) {\includegraphics[width=0.31\textwidth]{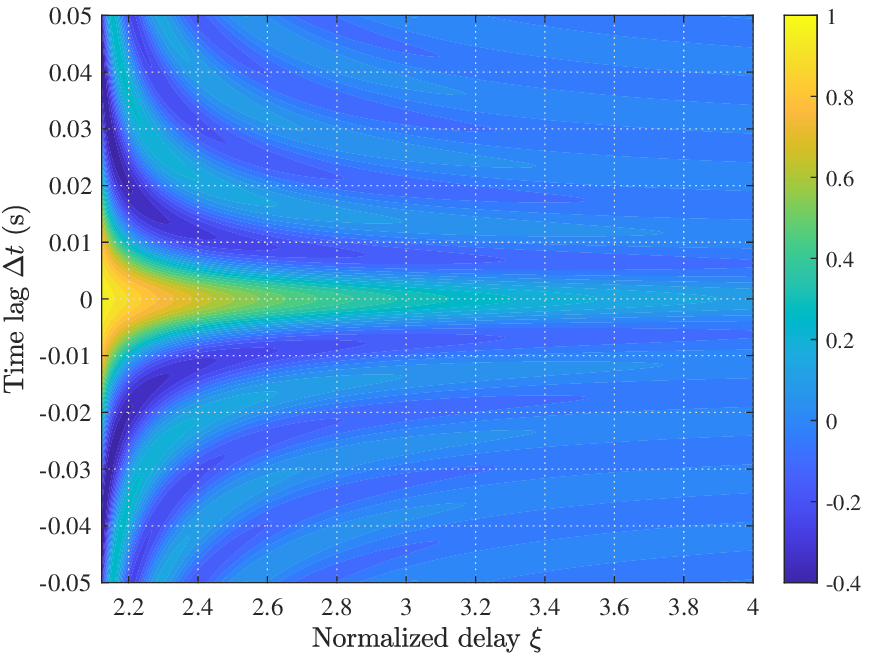}};
\node at ($(Shh)+(15pt,25pt)$) [draw=black, thin, inner sep=2pt, fill=white,minimum width=1.9cm,minimum height=0.9,align=left,font=\small] 
{$\rho(t; \xi, \Delta t)$};  

\node[label=above left:{(I)}] (Sss) at ($(center)+(-\columnWitdh/2+\blockWitdh/2,0)$) {\includegraphics[width=0.31\textwidth]{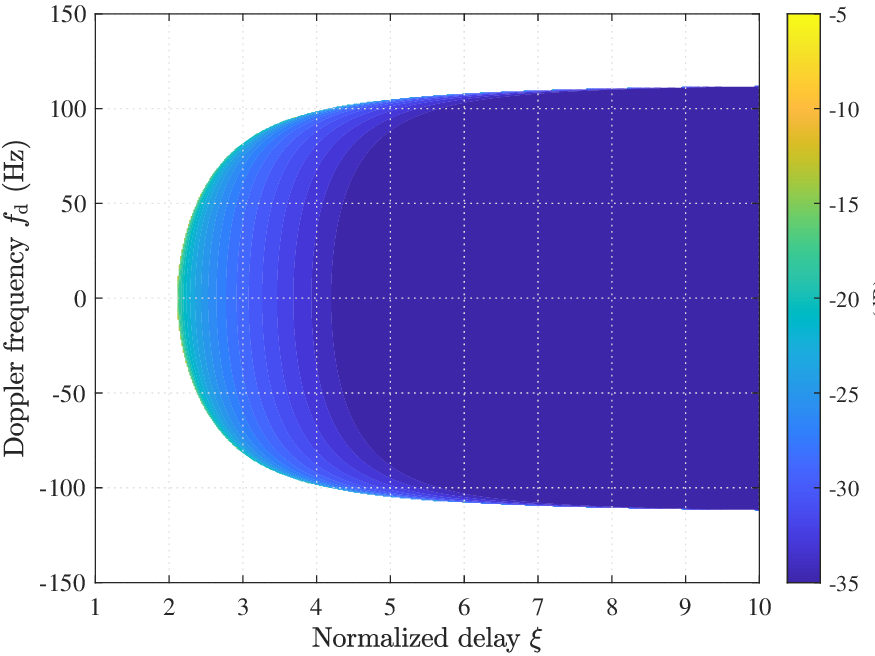}};
\node at ($(Sss)+(15pt,25pt)$) [draw=black, thin, inner sep=2pt, fill=white,minimum width=1.9cm,minimum height=0.9,align=left,font=\small] 
{$p(t; \xi,\fd)$};  

\node[label=above right:{(IV)}] (SHH) at ($(center)+(\columnWitdh/2-\blockWitdh/3.2,0)$) {\includegraphics[width=0.31\textwidth]{figures/time_freq_corr_theo.eps}};
\node at ($(SHH)+(15pt,25pt)$) [draw=black, thin, inner sep=2pt, fill=white,minimum width=1.9cm,minimum height=0.9,align=left,font=\small] 
{$r(t; \Delta \tilde{f}, \Delta t)$};  

\node[label=above left:{(i)}] (Pxi) at ($(Sss)+(0,2*\vertDist+0.22in)$) {\includegraphics[width=0.31\textwidth]{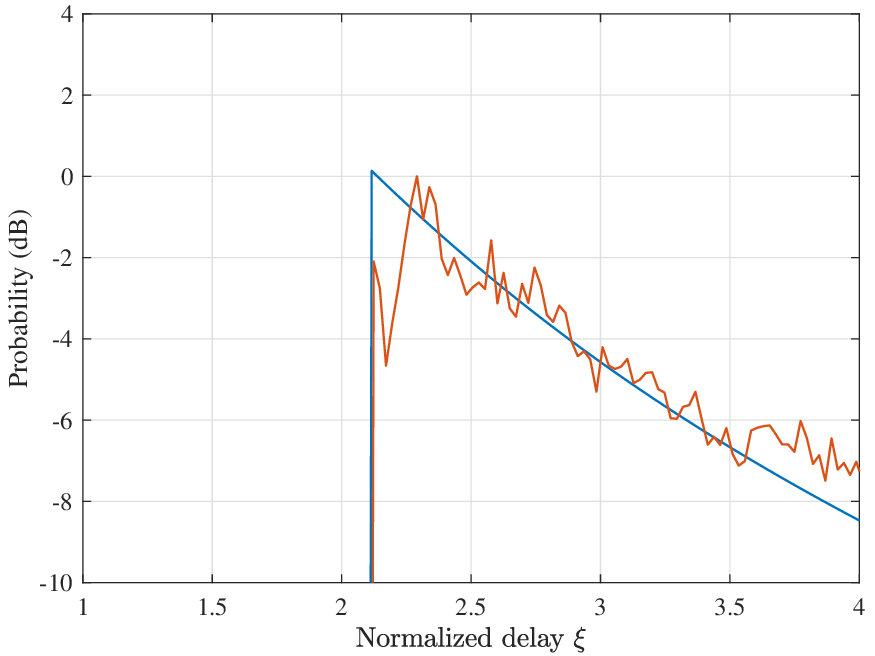}};
\node at ($(Pxi)+(22.5pt,21pt)$) [draw=black, thin, inner sep=1pt, fill=white,minimum width=2.2cm,minimum height=0.9cm,align=left,font=\small] 
{\hspace{-3pt}\mytikzLineMatlabBlue  \hspace{1pt} $p(t; \xi)$  \\ 
 \hspace{-3pt}\mytikzLineMatlabOrange \hspace{1pt} $\tilde{P}_h(t;\xi)$ };

\node[label=above left:{(ii)}] (Pfd) at ($(Sss)+(0,-2*\vertDist-0.22in)$) {\includegraphics[width=0.31\textwidth]{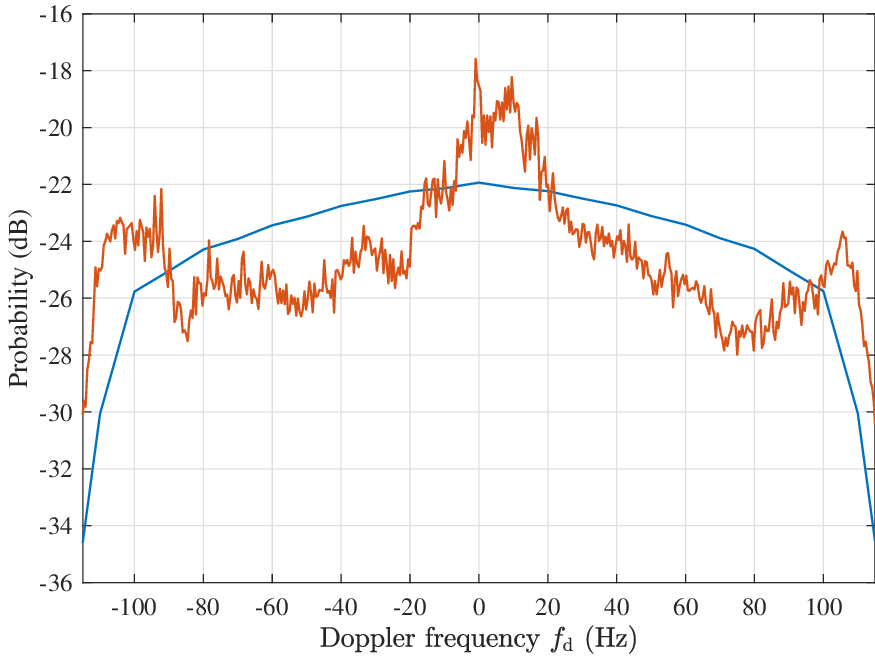}};
\node at ($(Pfd)+(26.7pt,22pt)$) [draw=black, thin, inner sep=1pt, fill=white,minimum width=1.9cm,minimum height=0.9,align=left,font=\small] 
{\mytikzLineMatlabBlue  \hspace{1pt} $p(t; \fd)$ \\ 
 \mytikzLineMatlabOrange \hspace{1pt} $\tilde{P}_{\varrho}(t;\fd)$};

\node[label=above right:{(iv)}] (Stt) at ($(SHH)+(0,-2*\vertDist-0.22in)$) {\includegraphics[width=0.31\textwidth]{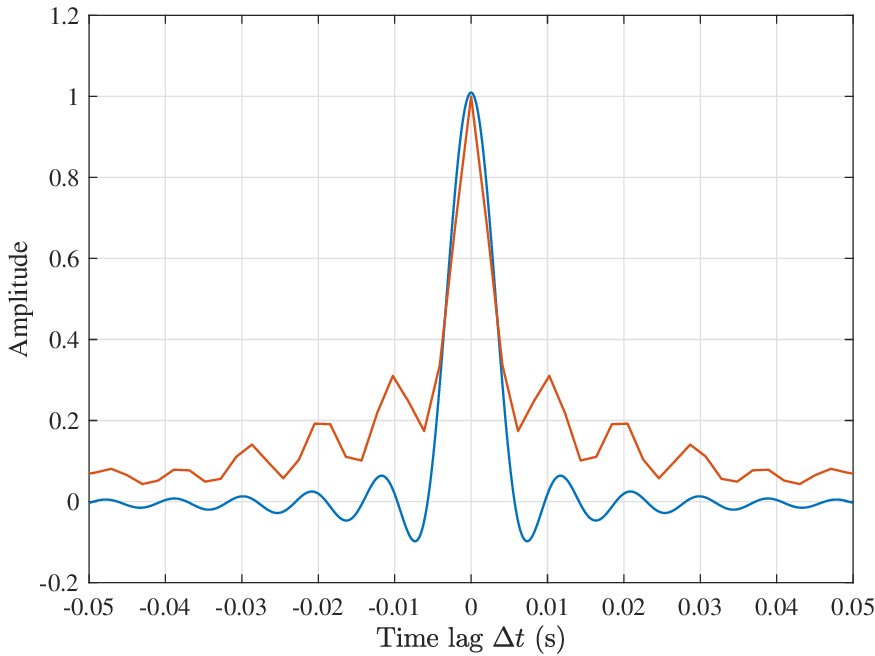}};
\node at ($(Stt)+(22.3pt,22pt)$) [draw=black, thin, inner sep=1pt, fill=white,minimum width=2.2cm,minimum height=0.9,align=left,font=\small] 
{\mytikzLineMatlabBlue  \hspace{1pt} $r(t; \Delta t)$ \\ 
 \mytikzLineMatlabOrange \hspace{1pt} $\tilde{R}_L(t;\Delta t)$};

\node[label=above right:{(iii)}] (Sff) at ($(SHH)+(0,2*\vertDist+0.22in)$) {\includegraphics[width=0.31\textwidth]{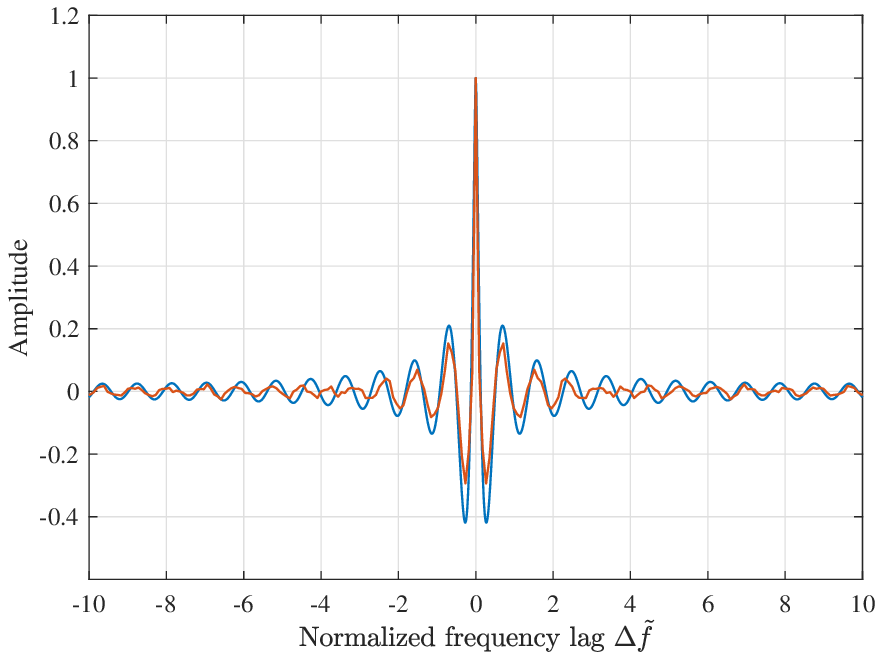}};
\node at ($(Sff)+(22.5pt,21pt)$) [draw=black, thin, inner sep=1pt, fill=white,minimum width=2.2cm,minimum height=0.9,align=left,font=\small] 
{\mytikzLineMatlabBlue  \hspace{1pt} $r(t; \Delta \tilde{f})$ \\ 
 \mytikzLineMatlabOrange \hspace{1pt} $\tilde{R}_L(t;\Delta \tilde{f})$};
 
\draw[*-o] (STT.west) -- (Sss.south);
\draw[*-o] (STT.east) -- (SHH.south);
\draw[*-o] (SHH.north) -- (Shh.east);
\draw[*-o] (Sss.north) -- (Shh.west);
\draw[*-o] (Sff.west) -- (Pxi.east);
\draw[*-o] (Pfd.east) -- (Stt.west);

\node[font=\footnotesize, below left] at (STT.west) {$\Delta \tilde{f}$};
\node[font=\footnotesize, below left] at (Sss.south) {$\xi$};
\node[font=\footnotesize, below right] at (STT.east) {$\fd$};
\node[font=\footnotesize, below right] at (SHH.south) {$\Delta t$};
\node[font=\footnotesize, above left] at (Sss.north) {$\fd$};
\node[font=\footnotesize, above left] at (Shh.west) {$\Delta t$};
\node[font=\footnotesize, above right] at (SHH.north) {$\Delta \tilde{f}$};
\node[font=\footnotesize, above right] at (Shh.east) {$\xi$};
\node[font=\footnotesize, above left] at (Sff.west) {$\Delta \tilde{f}$};
\node[font=\footnotesize, above right] at (Pxi.east) {$\xi$};
\node[font=\footnotesize, below right] at (Pfd.east) {$\fd$};
\node[font=\footnotesize, below left] at (Stt.west) {$\Delta t$};

\draw[-{Latex[scale=1.2]},dashed] ($(Sss.south) + (0,-2*\trafoSep)$) -- ($(Pfd.north) + (0,\trafoSep)$)node [midway,left] {\footnotesize $\displaystyle\int \hspace{-2.0pt} \boldsymbol{\cdot} \hspace{1.0pt} \mathrm{d}\xi\,$};
\draw[-{Latex[scale=1.2]}] ($(STT.south west) + (-\trafoSep,\trafoSep)$) -- ($(Pfd.north) + (0.5*\trafoSep,\trafoSep)$)node [midway,below right] {\footnotesize $\Delta \tilde{f}=0$};
\draw[-{Latex[scale=1.2]}] ($(SHH.south) + (0,-2*\trafoSep)$) -- ($(Stt.north) + (0,\trafoSep)$)node [midway,right] {\footnotesize $\Delta \tilde{f}=0$};
\draw[-{Latex[scale=1.2]}] ($(SHH.north) + (0,2*\trafoSep)$) -- ($(Sff.south) + (0,-\trafoSep)$)node [midway,right] {\footnotesize $\Delta t=0$};
\draw[-{Latex[scale=1.2]}] ($(Shh.north west) + (0,\trafoSep)$) -- ($(Pxi.south) + (0.5*\trafoSep,-\trafoSep)$)node [midway,above right] {\footnotesize $\Delta t=0$};
\draw[-{Latex[scale=1.2]},dashed] ($(Sss.north) + (0,2*\trafoSep)$) -- ($(Pxi.south) + (0,-\trafoSep)$)node [midway,left] {\footnotesize $\displaystyle\int \hspace{-2.0pt} \boldsymbol{\cdot} \hspace{1.0pt} \mathrm{d}\fd\,$};
\draw[-{Latex[scale=1.2]},dashed] ($(STT.north east) + (-\trafoSep,\trafoSep)$) -- ($(Sff.south) + (-0.5*\trafoSep,-\trafoSep)$)node [near end,above left] {\footnotesize $\displaystyle\int \hspace{-2.0pt} \boldsymbol{\cdot} \hspace{1.0pt} \mathrm{d}\fd\,$};
\draw[-{Latex[scale=1.2]},dashed] ($(Shh.south east) + (-\trafoSep,-\trafoSep)$) -- ($(Stt.north) + (-0.5*\trafoSep,\trafoSep)$)node [near end, below left] {\footnotesize $\displaystyle\int \hspace{-2.0pt} \boldsymbol{\cdot} \hspace{1.0pt} \mathrm{d}\xi\,$};
\end{tikzpicture}  